\documentclass[11pt]{article}
\usepackage{authblk}
\usepackage{hyperref}
\hypersetup{
backref=true
bookmarksnumbered,
pdfstartview={FitH},
citecolor={blue},
linkcolor={blue},
urlcolor={black},
pdfpagemode={UseOutlines}
}
\usepackage{doi}

\title{A game-theoretic mechanism for aggregation and dispersal of interacting populations}

\author{Russ deForest\thanks{email:
    \href{mailto:russ.f.deforest@gmail.com}{russ.f.deforest@gmail.com},
    corresponding author} }
\author{Andrew Belmonte\thanks{email: \href{mailto:andrew.belmonte@gmail.com}{andrew.belmonte@gmail.com}}}%
\affil{Department of Mathematics  \\ Pennsylvania State University,\\
    University Park, PA 16802 USA}

\usepackage[numbers]{natbib}
\usepackage{pgfplots}
\pgfplotsset{compat=newest}
\usepgfplotslibrary{groupplots}
\usetikzlibrary{spy}
\newlength\figwidth

\usepackage{amsmath,amssymb,amsthm}

\newtheorem{lem}{Lemma}
\newtheorem{thm}{Theorem}
\newtheorem{cor}{Corollary}

\theoremstyle{definition}
\newtheorem{defn}{Definition}
\newtheorem*{notation}{Notation}

\theoremstyle{remark}
\newtheorem*{remark}{Remark}

\numberwithin{equation}{section}

\newcommand{\abs}[1]{\left\lvert #1 \right\rvert}
\newcommand{\norm}[1]{\left\lVert #1 \right\rVert}
\newcommand{\R}{\mathbb{R}}

\newcommand{\dfit}{\delta \! f}
\newcommand{\pin}{\text{ in }}
\newcommand{\pon}{\text{ on }}
\newcommand{\real}{\textrm{Re }}
\newcommand{\bu}{\ensuremath{\mathbf{u}}}
\newcommand{\bv}{\ensuremath{\mathbf{v}}}
\newcommand{\f}{\ensuremath{\mathbf{f}}}

\newcommand{\Rd}{\ensuremath{\mathbb{R}^d}}

\begin{document}
\maketitle

\begin{abstract}
  We adapt a fitness function from evolutionary game theory as a
  mechanism for aggregation and dispersal in a partial differential equation (PDE) model
  of two interacting populations, described by density functions $u$
  and $v$.  We consider a spatial model where individuals migrate up
  local fitness gradients, seeking out locations where their given
  traits are more advantageous.   The resulting system of fitness gradient
  equations is a degenerate system having spatially
  structured, smooth, steady state solutions
  characterized by constant fitness throughout the domain.  When populations are viewed as
  predator and prey, our model captures prey aggregation behavior
  consistent with Hamilton's selfish herd hypothesis.
  We also
  present weak steady state solutions in 1d that are continuous but in
  general not smooth everywhere, with an associated fitness that is
  discontinuous, piecewise constant. We give numerical examples of 
  solutions that evolve toward such weak steady
  states.  We also give an example of a spatial Lotka--Volterra model,
  where a fitness gradient flux creates 
  instabilities that lead to spatially structured steady states.  Our
  results also suggest that when fitness has some dependence on local
  interactions, a fitness-based dispersal mechanism may act to create
  spatial variation across a habitat.

  \textbf{keywords:} dispersal, aggregation, fitness gradient,
    degenerate diffusion, quasilinear pde, cross-diffusive
    instability, evolutionary game dynamics, migration.

% \subclass{Primary 92D25; Secondary 35K65}
% \keywords{dispersal, fitness gradient, degenerate
% quasilinear parabolic pde, cross-diffusive instability, evolutionary game dynamics, migration}

\end{abstract}

\section{Introduction}
\label{sec:intro}

An interesting problem in ecology is understanding
the aggregation behavior seen in some prey species in the presence of
predators. In some settings, such as bait balls of mackerel
in the open ocean, aggregation provides an easy target for large
predators even as the behavior diffuses the risk to each individual~\cite{Parrish2002}.
  Thus, aggregation can be viewed as a cooperative
  behavior.

In his 1971 paper, \emph{Geometry of the Selfish
  Herd}, Hamilton hypothesized that the aggregation behavior in a prey
species could arise
from the selection pressure of predators~\cite{Hamilton1971}.  Under
the \emph{selfish herd hypothesis} each animal seeks to minimize its  individual domain of danger in the
presence of a predator that may appear at a random location.
Hamilton proposed a nearest neighbor rule; by moving in the direction
of the
nearest neighbor, prey animals reduce their individual domains of
danger. There are two interesting shortcomings of this assumption,
noted by Hamilton in his original paper.  First, the rule tends to produce 
small isolated clusters instead of large aggregations. Another is that
an animal may temporarily increase
its domain of danger in its approach to its nearest neighbor.  A variety of movement rules have followed.  A review paper by
Morrell and James summarizes movement rules that have been considered
and analyzes their success in capturing aggregation behavior in
various settings~\cite{Morrell2007}.

In this paper we adapt a fitness function from evolutionary game
theory as a mechanism for aggregation in a predator-prey model. The
model seeks to capture transitory dynamics of the interacting populations as
each locates itself on the landscape relative to the
other. We have in mind aggregation phenomena such as bait balls
where species of prey fish densely pack themselves
together in the presence of predators~\cite{Parrish2002}.

Our modeling assumptions will be shown to encode 
assumptions of Hamilton's \emph{selfish herd hypothesis};  prey tend to
aggregate so as to maximize their population density relative to the
population of predators and the predators follow.  Steady states are
characterized by the condition where the relative frequency of the
two populations is constant.  Our model can also be viewed as capturing
the spatial dynamics of a public goods type game, where both a
cooperating and defecting population each increase their fitness by
locating themselves in regions where the relative frequency of
cooperators is higher.

As is standard in evolutionary games, the fitness in our model depends on the
relative frequency of each population~\cite{Hofbauer1998}.  However, we are
not modeling
selection dynamics among competing traits or strategies.  Instead we
model the spatial dynamics as each population tends to move up its local
fitness gradient. The resulting model is a degenerate quasilinear
system of partial differential equations which we refer to as a
\emph{fitness gradient flux system} of partial differential equations (PDE).   The most interesting
feature  of the model is the presence of a negative density-dependent
diffusion coefficient for the prey population.  Naturally, this is the
feature in the model
that gives rise to the aggregation phenomenon.  The prey aggregation
is moderated by the predator population which ``chases'' the prey.
We show that a
perturbation of the model yields a  \emph{normally parabolic} system,
having smooth solutions.  Thus, despite the negative
self-diffusion coefficient, the degeneracy should be viewed as a
limiting case of a well-behaved system.  We also note that the mechanism
for aggregation differs from that used in the Keller-Segel chemotaxis
model  where aggregation follows
a chemical gradient and is moderated by
self-diffusion~\cite{Childress1981, Horstmann2003}.  

We present results on steady state solutions and on a linearization
around these steady states. We then discuss a % modified, 
% normally parabolic system motivated by the fitness gradient flux as a
Lotka--Volterra predator-prey model spatially extended via the
fitness-gradient flux.% resulting in a
% normally parabolic system with the prey population   again having a negative density dependent diffusion
% coefficient.

\subsection{The model}
\label{sec:model}

We consider
here 
a spatial model of two populations without selection, driven by
migration only in
the direction of increasing fitness, resulting from a \emph{fitness gradient flux}. This
flux arises naturally from the effort of individuals within each
population as they seek out, locally, positions of greater
advantage. The dynamics of the population densities can be modeled by  the following partial differential equations (PDEs) in a bounded domain $\Omega\subset \R^n$, with a no-flux boundary condition and strictly positive initial conditions
\begin{equation}   \label{eq:introPDE}     
   \begin{cases}
  \partial_t u = -\beta_1
  \nabla\cdot\left(u\nabla f_1\right), & \text{in } \Omega \times (0,T),\\
  \partial_t v  = -\beta_2 
  \nabla\cdot\left(v\nabla f_2\right), & \text{in } \Omega \times (0,T),\\
  \nu\cdot(u\nabla f_1) = 0, \quad \nu\cdot(v\nabla f_2) = 0,&  \text{on } \partial\Omega \times (0,T),\\
  u(\mathbf{x},0) = u_0(\mathbf{x})>0, & \text{in } \Omega,\\
  v(\mathbf{x},0) = v_0(\mathbf{x})>0, & \text{in } \Omega,
     \end{cases}
\end{equation} 
where $f_i$ describe the fitness for each population, the
$\beta_i$ are constants determining each population's sensitivity to its
fitness gradient, and $\nu$
is the outer unit normal to $\partial\Omega$. These equations, first
presented in~\cite{deforest2013}, describe
population migration in the direction of increasing fitness.  The resulting system can be viewed
as a generalized diffusion system, where there are cross-diffusion
effects (see Section~\ref{sec:recast}).

It is interesting to contrast a steady state solution of \eqref{eq:introPDE}
with an ideal free distribution.  In an
ideal free distribution a population is allocated to the available
habitat in an optimal way.  Fitness depends on local environmental
conditions and is assumed to be a decreasing function of the local
population density.   Constant fitness is a characteristic of ideal free distributions, since if fitness were not a constant function of
space, some individuals could relocate to more favorable habitat,
improving their fitness \cite{Fretwell1970, Cressman2006, Cosner2005}.  

As we show below in Section~\ref{sec:analysis}, for strictly positive steady state solutions  of \eqref{eq:introPDE} the fitness functions
$f_i$ are constant throughout $\Omega$.    However, our results differ
from an ideal free distribution in several ways.  In our model, an
individual's fitness is the expected value
of an interaction with another individual occupying the same local
area and  depends only on the ratio of the local population
densities, consistent with an evolutionary game.   As such, there are many
possible steady states giving  the same constant values for the fitness functions,
$f_i$.  Thus while populations at a steady state are optimally
distributed, there is no dependence on the background environmental
conditions, which are assumed to be uniform throughout the domain. The 
spatial structure of a particular steady state instead results from variations in the
ratio of population densities throughout the domain at some initial
time.  
 This suggests that when fitness has some dependence on intraspecies and
 interspecies interactions, fitness-based dispersal may act as a source
 of variation across a habitat.

In Section~\ref{sec:ode} we treat a simplified version of \eqref{eq:introPDE} on two nodes.  We derive this equation from a continuum limit argument in
Section~\ref{sec:continuum-limit}.  Our main results appear in
Section~\ref{sec:analysis}, where we discuss steady state solutions, weak
steady state solutions, and show that smooth steady state solutions
are unstable.  Numerical examples are
discussed in  Section~\ref{sec:numerics}. In the remainder of this section we provide some background
on  PDE models of interacting populations.

\subsection{Background on PDE Population Models}
\label{sec:background}

Although the use of diffusion in a PDE model of population dynamics originated with
Fisher~\cite{Fisher1937,Bacaer2011}, Skellam is credited as the first
systematic treatment of diffusion in modeling the spread of biological
populations~\cite{Skellam1951,Aronson1985,Okubo2001,Cantrell2003}. He further suggested that
such models must account for attractive and repulsive forces that arise
from animal behavior \cite{Skellam1973}.  Okubo extended
Skellam's work along these lines by allowing a transition
probability in a  biased random walk to depend either on local
conditions at a present node, or conditions at neighboring nodes and
at intermediate locations~\cite{Okubo1986}.

The first use of fitness-based migration in a  PDE model of biological
populations seems to be by Shigesada, Kawasaki, and Teramoto,  whose
work (now called the SKT model)
formalized Morisita's theory of environmental density~\cite{Shigesada1979}.  This theory, based on Morisita's 
experimental work with antlions and observations of other species,
assumes that the suitability of a given habitat declines with an
increase in population density and can be thought of as a precursor
to the assumptions used in
define an ideal free distribution \cite{Morisita1971,Rosenzweig1985}.
The SKT model  includes ``the attractive force which induces directed movements of individuals
toward favorable places'', as well as   random movements
(diffusion) and a 
nonlinear dispersive force due to population pressure. Their model demonstrates that dispersal due to population
pressure can reduce interspecific competition by leading
competitors to segregate spatially.

More recently, Cosner and Cantrell have used a fitness gradient flux in dynamic
models whose steady state solutions approximate ideal free
distributions~\cite{Cosner2005, Cantrell2008}. These are reaction-advection-diffusion models where the
advective term represents directed movement up a local fitness
gradient. The fitness is defined to be a local rate of reproduction and
is a decreasing function of the local population density. The key
result is that such a local dispersal mechanism can lead to an ideal
free distribution.  An extension to a two-species competition model
has been used to show that a species adopting a fitness-based
dispersal cannot be invaded by a competitor using only random
dispersal~\cite{Cantrell2013}.

A model of ideally-motivated
competitors was investigated in~\cite{Rowell2010},
demonstrating conditions for coexistence, spatial segregation, and
competitive exclusion.
A more recent paper by Cosner gives a thorough review of
the use of reaction-diffusion-advection models in studying both the
effects and evolution of  dispersal as well as providing background on
relevant analytical techniques~\cite{Cosner2014}. 

The Keller-Segel chemotaxis model is a well studied system modeling
aggregation. See the review in~\cite{Horstmann2003}. More recent work in
aggregation-diffusion equations has been focused degenerate on diffusion
and on aggregation with
nonlocal effects incorporated via convolution with a
potential~\cite{Bedrossian2011, Bertozzi2010}. An interesting
model in~\cite{Laurent2007} uses convolution with a smooth potential to regularize a density
dependent backward heat equation.

\subsection{Background on Normally Parabolic Reaction Diffusion
  Systems}

Our main results concern a
degenerate system of equations with a negative, density dependent
self-diffusion coefficient for one of the populations. Strong
solutions for an approximation to this  system, regularized by additional
diffusion terms were  shown to exist in~\cite{Xu2017}.

Under a different regularization, this
system may also be viewed as a limiting case of a normally parabolic
reaction diffusion system. The theory of quasilinear normally
parabolic systems is developed in a series of papers by
Amann~\cite{Amann1988b,AmannQ2, Amann1989b}. Such systems feature
spatial operators that are normally elliptic, but in general are not strongly
elliptic and they capture the smoothing property associated with the
heat equation.  That is to say such operators are generators of analytic semigroups
which can be used to represent solutions in an appropriate function
space. We state here a general, local in time existence result due to
Amann for
normally parabolic reaction-diffusion systems having a no-flux
boundary. In the sequel we show that a regularization of our model
results in such a system.  While in this model we focus on accessible features of the
degenerate system, we are interested in further study of normally
parabolic reaction-diffusion systems that retain key features of the
present model.

Let $\Omega$ be an open, bounded, and connected domain in $\R^n$ with  $C^2$ boundary
$\partial \Omega$.  We consider a system of PDEs acting on
real-valued functions
$\mathbf{u} = (u_1,\dots,u_d)$ given by
\begin{equation}
  \label{eq:indexPDE}
  \begin{aligned}
    \partial_t u_k &= \nabla \cdot (b_{k1}\nabla u_1 + \cdots +
    b_{kd}\nabla u_d) + f_k(\mathbf{u}), \quad \text{ for } k=1,\dots,d,\\
  \end{aligned}
\end{equation}
satisfying the no-flux boundary condition
\begin{equation}\label{eq:indexBC}
  \sum_{k=1}^d b_{kj}\nu \cdot\nabla u_j = 0,
\end{equation}
where $\nu$ is the outer unit normal to $\Omega$.  
Define
\begin{equation*}
  G =: \{\xi\in \R^d : \sigma(B(\xi)) \subset [\real z > 0]\}.
\end{equation*}
We take the coefficient functions $b_{kj}(\cdot)$ and the ``reaction''
functions $f_k(\cdot)$ to be smooth maps from $G$ to $\R$:
\begin{equation}\label{eq:coeff_condition}
b_{kj}\in C^\infty(G),\quad f_k\in C^\infty(G).
\end{equation}
\begin{remark}
  The set $G$ is open in $\Rd$. For the PDE we consider we will want
  \begin{equation*}
    R_{\geq 0}^d = \{x\in \Rd: x_i\geq 0 \text{ for all } i\} \subset G,
  \end{equation*}
or at least
\begin{equation*}
  R_{>0}^d = \{x\in \Rd: x_i>0 \text{ for all } i\} \subset G.
\end{equation*}

\end{remark}

Expressing the spatial differential operator in \eqref{eq:indexPDE} in
the form 
\begin{equation*}
  \mathcal{B}(\mathbf{u})\mathbf{u} = \nabla \cdot
  \left[B(\mathbf{u})\nabla \mathbf{u}\right],
\end{equation*}

and the boundary operator in \eqref{eq:indexBC} in the form
\begin{equation*}
  \mathcal{C}(\mathbf{u})\mathbf{u} = 0, 
\end{equation*}
we can rewrite our PDE in the form
\begin{equation}
  \label{eq:2}
  \begin{aligned}
    \partial_t \mathbf{u} = \mathcal{B}(\mathbf{u})\mathbf{u} +
    \mathbf{f}(\mathbf{u}),\\
    \mathcal{C}(\mathbf{u})\mathbf{u}= 0
  \end{aligned}
\end{equation}
Finally, let $V$ denote the space of $\Rd$-valued functions in $W^{1,p}(\Omega)$
that take values in $G$:
\begin{equation*}
V :=\left\{\bv \in W^{1,p}(\Omega): \bv(\overline{\Omega})\subset G\right\}.
\end{equation*}

We use the following local existence theorem \cite[pg. 17]{AmannQ2}.

\begin{thm}[Amann Normally Parabolic Local Existence]\label{thm:amann-local-existence}

For any $\bu_0\in V$, the PDE given by
  \begin{equation}
    \label{eq:AbstractPDE}
    \begin{cases}
      \partial_t \mathbf{u} = B(\mathbf{u})\mathbf{u} + \f(\mathbf{u}), & \pin \Omega \times
      (0,T),\\
      C(\mathbf{u}))\mathbf{u} = 0, & \pon \partial \Omega \times (0,T),\\
      \mathbf{u}(\cdot,0) = \mathbf{u}_0, & \pon \Omega,
    \end{cases}
  \end{equation}
satisfying (assumptions above) has a unique maximal solution,
\begin{equation}
  \label{eq:NP-soln}
  \bu(\cdot) \in C([0,t_f))\cap C^\infty(\overline{\Omega}\times (0,t_f),
  \R^d).
\end{equation}
The  map $t\mapsto \bu(t)$ defines a smooth semiflow on
$V$, in the
$W^{1,p}(\Omega)$ sense.  Furthermore, if $\bu(t)$ is a bounded orbit that is also bounded
away from the boundary $\partial V$ then $\bu(t)$ is relatively compact in $V$ and
for $t>0$ is 
also bounded in $W^{2,p}(\Omega)$. 
\end{thm}

\begin{remark}
  In general, $t_f$ depends on the initial condition $\bu_0$. If
  $\bu(t)$ remains bounded in $W^{1,p}(\Omega)$ and bounded away from
  $\partial V$ then we may take $t_f = \infty$; existence is
  global. 
\end{remark}

%%% Local Variables: 
%%% mode: latex
%%% TeX-master: "GameTheoreticAggregationMain"
%%% End: 

% \section{A Fitness Gradient Flux Predator-Prey Model}
% \label{sec:fitn-grad-model}

% As discussed in Section~\ref{sec:introduction} our model for two
% populations can be interpreted as a predator-prey model, describing
% transient dynamics and aggregation of prey as the two populations seek
% to locate themselves optimally on a landscape relative to one
% another.  We  consider the following system modeling the movement
% of two populations in a bounded region $\Omega\subset \R^n$:
% \begin{equation}   \label{eq:FFPDE1}   
% \begin{cases}  
% \partial_t u = -\nabla \cdot (u \nabla f), &  \pin \Omega
% \times (0,T)\\
%   \partial_t v  = -(1+\gamma)   \nabla\cdot\left(v\nabla f\right), & \pin \Omega \times (0,T),\\
%   \nu\cdot(u\nabla f) = \nu\cdot(v\nabla f) = 0,&  \text{on } \partial\Omega \times (0,T),\\
%   u(x,0) = u_0(x)>0, & \text{in } \Omega,\\
%   v(x,0) = v_0(x)>0, & \text{in } \Omega,
%      \end{cases}
% \end{equation} 
% The populations $u$ and $v$ each tend to move up the gradient of a density
% dependent fitness, derived from an evolutionary game. Under our assumptions the
% gradients of the individual fitness function are proportional; therefore there is
% an overall fitness gradient that each population tracks, with one
% population (the predator population $v$ under our assumption that $\gamma>0$) being more responsive to this
% gradient than the other.

\section{Recasting the Model  as a Generalized
  Diffusion System } 
 \label{sec:recast}

Here we calculate fitness gradients based on fitness functions
from an evolutionary game
for two populations and recast \eqref{eq:introPDE}
as a generalized degenerate diffusion system.  We also demonstrate
a regularizing perturbation that results in a normally parabolic
quasilinear system.

We first consider the following system, first presented
in~\cite{deforest2013}, where we've denoted the
fitness for $u$ as $f_1$; the fitness for $v$ is denoted by $f_2$:
\begin{equation}   \label{eq:FFPDE2}   
\begin{cases}  
\partial_t u = -\beta_1\nabla \cdot (u \nabla f_1), &  \pin \Omega
\times (0,T)\\
  \partial_t v  = -  \beta_2\nabla\cdot\left(v\nabla f_2\right), & \pin \Omega \times (0,T),\\
  % \nu\cdot(u\nabla f) = \nu\cdot(v\nabla f) = 0,&  \text{on } \partial\Omega \times (0,T),\\
  % u(x,0) = u_0(x)>0, & \text{in } \Omega,\\
  % v(x,0) = v_0(x)>0, & \text{in } \Omega,
     \end{cases}
   \end{equation}
   The constants $\beta_i$ denote each population's responsiveness or
   sensitivity to its individual fitness gradient.  By letting
   $\beta=\frac{\beta_2}{\beta_1}$ and rescaling time, we may re-write
   \eqref{eq:FFPDE2} as
 \begin{equation}   \label{eq:FFPDE3}   
\begin{cases}  
\partial_t u = -\nabla \cdot (u \nabla f_1), &  \pin \Omega
\times (0,T)\\
  \partial_t v  = -\beta  \nabla\cdot\left(v\nabla f_2\right), & \pin \Omega \times (0,T),\\
     \end{cases}
\end{equation}

Now let $A$ denote a two strategy
symmetric game matrix,
\begin{equation*}
  A =
  \begin{bmatrix}
    a_{11} & a_{12}\\a_{21}& a_{22}
  \end{bmatrix}.
\end{equation*}
The fitness functions for $u$ and $v$ are
\begin{equation}\label{eq:fitness}
  f_1(u,v) = \frac{a_{11}u + a_{12}v}{u+v}, \quad f_2(u,v) =
  \frac{a_{21}u + a_{22}v}{u+v}.
\end{equation}
This definition of
fitness, which we base on Taylor and Jonker~\cite{Taylor1978}, is standard in
the evolutionary game literature (see also~\cite{Vickers1989},
\cite{Demetrius2000}, \cite[Ch. 7]{Easley2010},\cite{Hofbauer1998}).

This leads to
\begin{equation}
  \label{eq:fitness-gradient}
  \begin{aligned}
    \nabla f_1(u,v) &= \frac{(a_{11}-a_{12})}{(u+v)^2}\left(v\nabla u -
    u\nabla v\right)\\
\nabla f_2(u,v) & = \frac{(a_{21}-a_{22})}{(u+v)^2} \left(v\nabla u -
  u\nabla v\right).
  \end{aligned}
\end{equation}
Noting that the fitness gradients are proportional, we define a
constant depending on the matrix $A$,
\begin{equation*}
  \kappa_A = \frac{(a_{21}-a_{22})}{(a_{11}-a_{12})}, 
\end{equation*}
so that
\begin{equation*}
  \nabla f_2 = \kappa_A \nabla f_1.
\end{equation*}
As we have done in \eqref{eq:introPDE}, we will usually denote the
fitness function for $u$ as $f(u,v)$. The corresponding fitness gradient for
$v$ is then
\begin{equation*}
  \kappa_A\nabla f. 
\end{equation*}
By defining a parameter $\gamma>0$ such that
\begin{equation*}
  (1+\gamma) = \kappa_A,
\end{equation*}
we arrive at the following  PDE system.
\begin{equation}   \label{eq:FFPDE1}   
\begin{cases}  
\partial_t u = -\nabla \cdot (u \nabla f), &  \pin \Omega
\times (0,T)\\
  \partial_t v  = -(1+\gamma)   \nabla\cdot\left(v\nabla f\right), & \pin \Omega \times (0,T),\\
  \nu\cdot(u\nabla f) = \nu\cdot(v\nabla f) = 0,&  \text{on } \partial\Omega \times (0,T),\\
  u(x,0) = u_0(x)>0, & \text{in } \Omega,\\
  v(x,0) = v_0(x)>0, & \text{in } \Omega.
     \end{cases}
\end{equation} 

We normalize the
game matrix, making the assumption that $a_{11}-a_{12}=1$; we also
assume that $\gamma>0$ which requires that $a_{21}-a_{22}>1$. Thus $v$ is more responsive to $\nabla f$ than the
population $u$. 

\begin{remark}
  This game dynamic is one where players of type $u$ do best
against their own type, while individuals of type $v$ do  better
against type $u$ than against others of their own type.  This is the case in the
classical prisoner's dilemma and hawk-dove games.  As discussed in the
introduction, a
similar dynamic exists in  predator-prey systems where prey aggregate
and by so doing, reduce their individual risk of predation 
while predators benefit by locating
themselves where prey is highly concentrated \cite{Hamilton1971, Ruxton2006,Wrona1991}. 
We are not considering
 Nash equilibria of the game given by $A$; our interest here is only in how the
game matrix $A$, though the parameter
$\gamma$, affects the movement and distribution of the populations.  

\end{remark}

% We make assumptions
% on the game matrix $A$ such that
% $\gamma>0$; thus $v$ is more responsive to $\nabla f$ than the
% population $u$.  We summarize these assumptions here:
% \begin{enumerate}
% \item $a_{11}-a_{12}=1$ (normalization of game matrix)
% \item $a_{21}-a_{22}>1$ 
% \end{enumerate}

From \eqref{eq:fitness-gradient} with
$a_{11}-a_{12} = 1$ we have
\begin{equation*}
  \begin{aligned}
    -u\nabla f &= \frac{1}{(u+v)^2}\left(-uv\nabla u + u^2\nabla
      v\right)\\
  -(1+\gamma) \nabla f & = \frac{(1+\gamma)}{(u+v)^2}\left(-v^2\nabla
    v + uv \nabla v\right)
  \end{aligned}
\end{equation*}
It is sometimes convenient to express \eqref{eq:FFPDE1} in the
vector form
\begin{equation}\label{vector-form}
  \partial_t \mathbf{w} = \nabla \cdot \left(b(\mathbf{w}) \nabla \mathbf{w}\right),
\end{equation}
where $\mathbf{w} = (u,v)$ and $\nabla \mathbf{w} = (\nabla u, \nabla
v)$, where $B(\cdot)$ a matrix of density dependent diffusion coefficients with
\begin{equation}
  \label{eq:mB1}
  B(u,v) = \frac{1}{(u+v)^2}
  \begin{bmatrix}
    -uv & u^2\\ -(1+\gamma)v^2 & (1+\gamma)uv
  \end{bmatrix}.
\end{equation}

\begin{remark}
 It is now clear  how our choice of fitness functions for the two populations
encodes the assumption that both the prey population $u$ and the
predator population $v$ tend toward regions where the population $u$
has higher density.  The term $-uv \nabla u$ in the flux for $u$
drives the prey aggregation, while the term $u^2 \nabla v$ indicates
that prey are also seeking to move away from higher concentrations  of the predators.  Meanwhile the term $-v^2\nabla v$ indicates the predators
are chasing the prey, while the term $uv\nabla v$ indicates some
intraspecies competition among predators. 
\end{remark}

Note that for $(u,v)\in \R^2_{>0}$ the matrix $B(u,v)$ has one zero
eigenvalue and one positive eigenvalue. 
% When $u$ and $v$ are strictly positive, $B(\cdot)$ has one positive
% eigenvalue and one zero eigenvalue.

% \todo{Do this!}
% [\textbf{Helpful if we redefine fitness $f$ so that $\nabla f$  = 0 at
%   origin.  Then PDE defined on $[0,\infty)^2$}]

\begin{lem}
   The eigenvalues of $B(u,v)$ are
   \begin{equation*}
     \lambda_1 = 0,\quad \lambda_2 = \frac{\gamma uv}{(u+v)^2}
   \end{equation*}
\end{lem}

Regularizing \eqref{eq:FFPDE1} by including the additional
diffusion terms $\epsilon \Delta u$ and $\epsilon \Delta v$ results in
a \emph{normally parabolic} diffusion system.

\begin{thm}\label{thm:npffpde1}
  For all $\epsilon>0$, the following PDE system is normally parabolic
  for strictly positive $u,v$. 
\begin{equation}   \label{eq:npffpde1}   
\begin{cases}  
\partial_t u = -\nabla \cdot (u \nabla f) + \epsilon \Delta u, &  \pin \Omega
\times (0,T)\\
  \partial_t v  = -(1+\gamma)   \nabla\cdot\left(v\nabla
    f\right)+\epsilon \Delta v, & \pin \Omega \times (0,T),\\
  \nu\cdot(u\nabla f-\epsilon \nabla u) = \nu\cdot(v\nabla f-\epsilon
  \nabla v) = 0,&  \text{on } \partial\Omega \times (0,T),\\
  u(x,0) = u_0(x)>0, & \text{in } \Omega,\\
  v(x,0) = v_0(x)>0, & \text{in } \Omega,
     \end{cases}
\end{equation} 
\end{thm}

\begin{proof}
  
 Note that  for the regularized PDE system in \eqref{eq:npffpde1}, the matrix
   of coefficients $B_\epsilon(u,v)$ is
   \begin{equation*}
     B_\epsilon = B(u,v) + \epsilon I
   \end{equation*}
having the eigenvalues $\lambda_1 = \epsilon$ and $\lambda_2 =
\frac{\gamma uv}{(u+v)^2} + \epsilon$. These eigenvalues are strictly
positive for $u,v\in R^2_{>0}$.

 The set
  \begin{equation*}
    G = \left\{\xi\in \R^2 : \sigma(B_\epsilon(\xi)) \subset [\real z
      >0]\right\}
  \end{equation*}
clearly contains $R^2_{>0}$. Since the PDE system also satisfies
(1.6)-(1.7) in this region, it follows that
\eqref{eq:npffpde1} is normally parabolic for strictly positive $u,v$.
\end{proof}

% \begin{proof}
%   This follows directly from \eqref{eq:mB1} and from 

% \begin{itemize}
% \item If $\lambda$ is eigenvalue of $M$, then $c\lambda$ is
%   eigenvalues of $cM$.

% \item If $\lambda$ is eigenvalue of $M$, then $\lambda + \epsilon$ is
%   eigenvalue of $M+\epsilon I$
% \end{itemize}

% \end{proof}

%[Proof of Theorem~\ref{thm:npffpde1}]
 
\begin{cor}
  Given strictly positive initial conditions $u_0$ and $v_0$ and
  taking $\bu(t) = (u(t),v(t))$
  $C(\overline{\Omega})$, \eqref{eq:npffpde1} has a unique maximal solution
  \begin{equation}
    \label{eq:npffpde-solution}
    \bu(\cdot) \in C([0,t_f])\cap C^\infty(\overline{\Omega} \times
    (0,t_f), \Rd).
  \end{equation}
\end{cor}

\begin{proof}
  This follows directly from Theorem~\ref{thm:amann-local-existence}.
\end{proof}

\begin{remark}
  Note that $B(u,v)$ in \eqref{eq:mB1} is not symmetric and our
  resulting PDE is non-coercive. Nonetheless, under our
  assumptions,  the perturbed system involving $B_\epsilon(u,v)$
  satisfies the conditions to be normally parabolic having smooth, local in time solutions.
\end{remark}

%%% Local Variables: 
%%% mode: latex
%%% TeX-master: "GameTheoreticAggregationMain"
%%% End: 

\section{A simple system: two coupled spatial points}
\label{sec:ode}

It is instructive to consider a discrete spatial model for the
fitness-gradient flux, where population movements can be described by
a system of ordinary differential equations (ODE).
Here we consider two populations moving
between two nodes. In analogy to
the system described by \eqref{eq:FFPDE1},  movement of each population is
determined by a fitness gradient; simply put,  movement is toward
the node where population fitness is higher. As in Section~\ref{sec:recast}, fitness is
defined by the expected payoff of an underlying 
evolutionary game between the two populations, and we again imagine
the dynamic as movement of a prey population, with density given by
$u$ and a predator population, whose density is given by $v$. %  This definition of
% fitness, which we base on Taylor and Jonker \cite{Taylor1978} is standard in
% the evolutionary game literature (see also \cite{Vickers1989},
% \cite{Demetrius2000}, \cite[Ch. 7]{Easley2010},\cite{Hofbauer1998}).

% \todo{slight re-write to fit with re-organization}
% In this section we consider a discrete spatial model for two
% populations, denoted by $u$ and $v$, on two spatial nodes and described by a system of ordinary
% differential equations (ODEs). The movement of each population is
% determined by a fitness gradient; simply put,  movement is toward
% the node where population fitness is higher, where fitness is
% defined by the expected payoff of an underlying 
% evolutionary game between the two populations.  This definition of
% fitness, which we base on Taylor and Jonker \cite{Taylor1978} is standard in
% the evolutionary game literature (see also \cite{Vickers1989},
% \cite{Demetrius2000}, \cite[Ch. 7]{Easley2010},\cite{Hofbauer1998}).

In this context, fitness depends only on the population ratio
$u/v$ at each node. A steady state is reached when these ratios are equal
between nodes, or when both populations accumulate at a single node,
leaving the other node empty. In other words, the prey either
distributes its population between nodes so that it is in constant
ratio to the predator population, or the entire prey population
aggregates to a single node, followed by

As we show, the system  approaches a steady state
for any initial conditions and for particular initial conditions, both populations 
accumulate on a single node.  This is perhaps the most interesting
behavior of this basic model as it provides some insight into the pinching off
behavior observed the fitness-gradient flux PDE system for two
populations given by~\eqref{eq:FFPDE1}.  

% We follow these results with a  derivation of \ref{eq:mainPDE} via a
% continuum limit argument.

% Prior to analyzing the PDE system given in
% \eqref{eq:Fitness-Flux-System}, we consider

%\subsection{Discrete Model and Basic Results}

For $i=1,2$, let $u_i(t)$ and $v_i(t)$ denote populations at node $i$
at time $t\geq 0$.    In the model under consideration, population changes are due entirely
to migration  between nodes, as described by the following system of ODEs:
\begin{align}
  \label{eq:ODE1}
  \dot{u}_1 &=
  \begin{cases}
    u_2\left[f_1(u_1,v_1)-f_1(u_2,v_2)\right], & \text{if } f_1(u_1,v_1)\geq
    f_1(u_2,v_2),\\
    -u_1\left[f_1(u_2,v_2)-f_1(u_1,v_1)\right], & \text{if } f_1(u_1,v_2)< f_1(u_2,v_2),
  \end{cases}\\
  \label{eq:ODE2}
  \dot{v}_1 & = \begin{cases}
    \beta v_2\left[f_2(u_1,v_1)-f_2(u_2,v_2)\right], & \text{if } f_2(u_1,v_1)\geq
    f_2(u_2,v_2),\\
    -\beta v_1\left[f_2(u_2,v_2)-f_2(u_1,v_1)\right], & \text{if }
    f_2(u_1,v_2)< f_2(u_2,v_2),
  \end{cases}\\
\dot{u}_2 &= -\dot{u}_1,\\
\dot{v}_2 &= -\dot{v}_1.
\end{align}
The function $f_1(u,v)$ describes a fitness for $u$ that depends
only on the
relative size of the populations $u$ and $v$ (at a given node).
Similarly, $f_2(u,v)$ describes the fitness of $v$.  The parameter
$\beta>0$ indicates the degree to population $v$ is sensitivity 
to a difference in fitness, relative to ppopulation $u$'s sensitivity, as discussed in Section~\ref{sec:recast}.%   $\beta_i$ are
% positive parameters that indicate each population's sensitivity to a
% difference in fitness between the nodes.

For each population, migration between the nodes corresponds to
movement in the direction of increasing fitness. Fixed
points of the system occur when the fitness of both species is equal
between the two nodes. For our definition of fitness, this occurs when
the populations satisfy
\begin{equation*}
  u_1v_2 = u_2v_1.
\end{equation*}
This occurs when the population ratios at each node are equal or when
both populations accumulate at a single node, as discussed below.

% \subsection{Fitness definition and model simplification}
% \label{sec:fitn-defin-model}
% \todo{a lot of this now repeats fitness gradient section above }

As in \ref{sec:recast}, we take
% Consider a symmetric evolutionary game defined by the payoff matrix
% \begin{equation}
%   \label{payoff-matrix}
%   \mathbf{A} =
%   \begin{pmatrix}
%     a_{11} & a_{12}\\a_{21}& a_{22}
%   \end{pmatrix}.
% \end{equation}
% Following \cite{Vickers1989}  and \cite{Taylor1978}, we define the fitness for each population as the expected payoff for the corresponding strategy
\begin{equation}
  \label{eq:evol-fitness-defn}
  f_1(u,v) = \frac{a_{11}u+a_{12}v}{u+v}, \quad f_2(u,v) = \frac{a_{21}u+a_{22}v}{u+v},
\end{equation}
and set
\begin{equation}\label{eq:kappaA}
  a_{11}-a_{12}=1,\quad \kappa_A = a_{21}-a_{22}>0
\end{equation}

% We impose the following assumption on $A$:  
% \begin{equation}
%   \label{a:game-assumption} %assume a_{11}-a_{12}=1??
%      a_{21}-a_{22}>0,\quad a_{11}-a_{12}>0,
% \end{equation}
% and define a parameter
% \begin{equation}
%   \label{eq:kappa-defn}
%   \kappa_A = \frac{a_{21}-a_{22}}{a_{11}-a_{12}}.
% \end{equation}

Let $\dfit$ denote the difference in fitness for population $u$ between the two nodes, 
\begin{align*}
  \dfit &= f_1(u_1,v_1)-f_1(u_2,v_2).
\end{align*}
Then using \eqref{eq:evol-fitness-defn}, \eqref{eq:kappaA} we have
\begin{align}
\label{eq:dfit}
  \dfit &= \frac{1}{(u_1+v_1)(u_2+v_2)}\left(u_1v_2 -
    u_2v_1\right),\\%\intertext{and}
 \kappa_A\,\dfit & = f_2(u_1,v_1)-f_2(u_2,v_2),\nonumber
\end{align}
% so that the differences in fitness for each population are proportional.
% To  simplify our model we  rescale time by
% $\beta_1$, let $\beta=\beta_2/\beta_1$, 
Again as in Section~\ref{sec:recast}, we  make the assumption $\beta\kappa_A>1$,
and define the  
positive parameter $\gamma$ by
\begin{equation}
  \label{eq:gamma-defn}
  \gamma = \beta\kappa_A-1 > 0.
\end{equation}
% We discuss the details of this game dynamic in more detail in Section~\ref{sec:fitness}.
% \todo{More detail here? or reference elsewhere}
This allows us to re-write \eqref{eq:ODE1}-\eqref{eq:ODE2} in terms of
$\dfit$; effectively, population $v$ is more sensitive than $u$ to
differences in $f$ between the two nodes.  
\begin{align}
  \label{ODE-u}
    \dot{u}_1 &=
    \begin{cases}
      u_2\dfit, & \text{if } \dfit\geq0,\\
      u_1\dfit, & \text{if } \dfit<0,
    \end{cases}\\
    \dot{v}_1 &= 
    \begin{cases}
      (1+\gamma)v_2\dfit, & \text{if } \dfit\geq 0,\\
      (1+\gamma)v_1\dfit, & \text{if } \dfit<0,
    \end{cases}\label{ODE-v}\\
        \dot{u}_2 &= -\dot{u}_1,\\
        \dot{v}_2 &= -\dot{v}_1,\\
    u_i(0) &= u_i^0>0,\quad v_i(0) = v_i^0>0, \text{ for } i = 1,2.\label{ODE-ic}
\end{align}
\begin{remark}
In our reduced model, the assumption $\gamma>0$ implies that the
population $v$ is more sensitive to (or responds more rapidly to a
change in) the
difference in fitness, $\dfit$. This assumption is important for the
parabolic nature of the PDE discussed in Section~\ref{sec:recast}. 
\end{remark}

\begin{remark}
 From \eqref{eq:dfit}, we see that $u_1v_2=u_2v_1$ implies $\dfit=0$
 and hence $\dot{u}_i =
\dot{v}_i=0$. Motivated by this condition, we define $E(t) := u_1v_2-u_2v_1$, so
that the ODE system \eqref{ODE-u}-\eqref{ODE-ic} is at a steady state when $E=0$.  We  show that
$E\xrightarrow[]{t\to \infty} 0$, for any positive initial
conditions. Furthermore, for some initial conditions, the total
population accumulates at one node with the other node emptying out.

\end{remark}

Let us denote the total population at each node $i$ at time $t$ by
$S_i(t)$, 
  \begin{align*}
    S_1(t) &= u_1(t)+v_1(t),\\
    S_2(t) &= u_2(t)+v_2(t).
  \end{align*}
Because the populations $u_1+u_2$ and $v_1+v_2$ are conserved, there is
a bound on each $u_i$ and $v_i$ and hence on the product $S_1S_2$.
Therefore there exists a constant $M>0$ such that for all $t\geq0$,   
\begin{equation*}
S_1(t)S_2(t)\leq M.
\end{equation*}

%%%%%%%%%%%%%%%%%
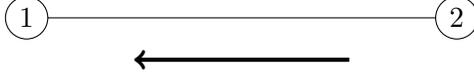
\begin{figure}[h]%[h]
  \centering
  \setlength{\figwidth}{\textwidth}
    \begin{tikzpicture}
    \begin{axis}[axis lines=none,width=4in, height=1.5in,
      xtick=\empty,
      ytick=\empty,
      xmin=-.25,
      xmax=1.25,
      ymin=-1,
      ymax=1
      ]
      \draw[] (axis cs: 0,0) circle (8pt) node[] {1} (axis cs: .05,0)
      -- (axis cs:.95,0)(axis cs: 1,0) circle (8pt)
      node[] {2};
      \draw[->, ultra thick] (axis cs: .75,-.5) -- (axis cs: .25,-.5);
%      \draw (axis cs: .5,-.5) node[above] {\small direction of flux};
%      \draw (axis cs: .5,-.5) node[below] {\small when $E>0$};
\end{axis}  
  \end{tikzpicture}

%%%%%%%%%%%%%%%%%%%%%%%%%%%%%%%%%%%%%%%%%%%%%%%%%%%%%%%

%%% Local Variables: 
%%% mode: latex
%%% TeX-master: "GameTheoreticDispersalMain"
%%% End: 
  \caption{Schematic of the two node model. The arrow denotes the direction of  flux for $u$ when $E>0$ (equivalently when $\dfit >0$).}
\label{fig:cell}
\end{figure}
%%%%%%%%%%%%%%%%%

% \begin{proof}
% Since the total population  $S_1(t)+S_2(t)$ is conserved, this is a straightforward consequence of the fact that the product $xy$ such that $x+y$ is constant is maximized when $x=y$. 

%   Without loss of generality, we assume that $\dfit>0$  (the flux is in the direction 
%   of node 1, as shown in Figure~\ref{fig:cell}), and $S_1(0)\geq
%   S_2(0)$ at $t=0$. Then
%   $u_1$ and $v_1$ are increasing, while $u_2,v_2$ are decreasing.  If
%   $S_1(0)=S_2(0)$, then we choose $M=S_1(0)S_2(0)$.  Otherwise, there is
%   some constant $a$ such that $S_1(0)-a = S_2(0)+a$ and then
%   \begin{equation*}
%     S_1(t)S_2(t) \leq M = (S_1(0)-a)(S_2(0)+a).
%   \end{equation*}
% \end{proof}

\begin{lem}\label{lem-ODE-ss-1}
  For any initial conditions $u_i(0),v_i(0)>0$ for $i=1,2$, the ODE system \eqref{ODE-u}-\eqref{ODE-ic} converges to a steady state such that $E =u_1v_2-u_2v_1 = 0$.
\end{lem}

\begin{proof}
  Again we assume that $\dfit$ is positive at $t=0$.   Then
  \begin{equation*}
    \begin{cases}
      \dot{u}_1 = u_2 \, \dfit, \\
      \dot{v}_1 = (1+\gamma) v_2  \, \dfit,\\
      \dot{u}_2 = -\dot{u}_1,\\
      \dot{v}_2 = -\dot{v}_1,
    \end{cases}
  \end{equation*}
and $E(0) = (u_1v_2 - u_2v_1)> 0$, since $\dfit >0$.
Next,
\begin{align*}
  \dot{E} & = \dot{u_1}v_2 + u_1\dot{v}_2- \dot{u}_2v_1 - u_2\dot{v}_1 \\
  & =  \dot{u}_1(v_1+v_2) - \dot{v_1}(u_1+u_2)\\
  & =  u_2(v_1+v_2) \dfit-  (1+\gamma)v_2(u_1+u_2)  \, \dfit\\
  & = [(u_2v_1-u_1v_2)-\gamma v_2(u_1+u_2)]  \, \dfit\\ 
  & \leq -E  \, \dfit \\
  & = -\frac{1}{S(t)T(t)}E^2 %\\
  %& 
  \, \leq \,\, -\frac{1}{M}E^2. 
\end{align*}
  Thus,
  \begin{equation*}
    \dot{E}\leq -CE^2
  \end{equation*}
    for some $C>0$, which implies that $E(t)\xrightarrow[]{t\to\infty}
  0.$ To see this define   \begin{equation*}
  F(t) = \frac{E(0)}{Ct+1},
\end{equation*}
and note that $F(0)=E(0)$, $\dot{E}\leq \dot{F}$, and
$F(t)\xrightarrow[]{t\to \infty} 0.$ 
\end{proof}

Since  $u_1$ and $v_1$ are initially increasing (for $E(0)>0$), we have in the limit
  \begin{equation*}
    \frac{u_1}{v_1} v_2 = u_2.
  \end{equation*}
If $v_2> E$ at $t=0$, then $v_2$ and $u_2$ remain bounded away from zero
and the steady state condition can also be written
$$
\frac{u_1}{v_1} = \frac{u_2}{v_2}.
$$
%is characterized by \eqref{ODE-ss}.  

When $v_2\leq E$ at $t=0$, however, we will show that node 2 empties out, as the entire population moves to node 1.
Consider the projection of trajectories to the
$(v_2,E)$-phase plane (see Figure \ref{fig:vE-plane}), for which
\begin{align}
  \label{phase_plane}
  \frac{dE}{dv_2} = \frac{\dot{E}}{\dot{v_2}} = \frac{E+\gamma
    v_2}{(1+\gamma)v_2} = \frac{1}{(1+\gamma)v_2}E + \frac{\gamma}{(1+\gamma)}.
\end{align}
Notice that when $E(t)=v_2(t)$, then
\begin{equation*}
  \frac{dE}{dv_2} = 1,
\end{equation*}
so that the trajectory remains along the line $v_2=E$, approaching the origin as $t\to \infty$. This line  divides the phase plane into two regions that characterize the asymptotic behavior.  Trajectories for which $E\geq v_2$ at $t=0$ (i.e. begin on or above the line $E = v_2$) will also approach the origin, while trajectories with $E<v_2$ at $t=0$ (beginning below the line) approach a positive value of $v_2$ along the $v_2$-axis - see Figure~\ref{fig:vE-plane}.

\begin{figure}[h]
  \centering
  \setlength{\figwidth}{.8\textwidth}
  \input{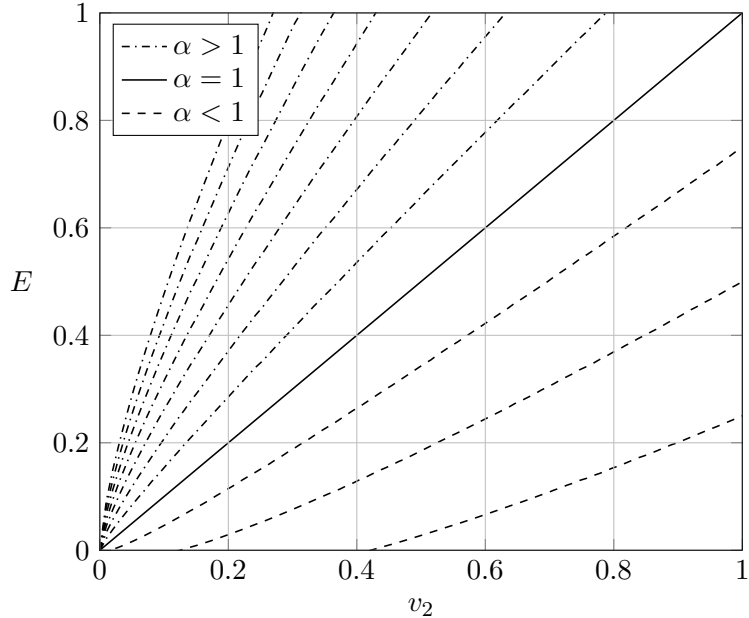}
  %\includegraphics[height=3in]{trajectories1}
  
% \hspace{-8mm}
%   \includegraphics[height=2.1in]{Trajectories-bdry-between-r1-r2-for-different-gamma} 
  \caption{ODE dynamics in the $(v_2,E)$-plane, for
    $\gamma=0.5$ and $E>0$ at $t=0$. Trajectories satisfy 
    \protect \eqref{traj-soln}, with $\alpha$ determined by the initial
      conditions ($\alpha = E(1)$).  For $\alpha>1$, trajectories approach the origin as
    $t\to \infty$.  For $\alpha=1$, the trajectory approaches the origin
    along the line $E=v_2$. Trajectories for $\alpha<1$ approach the point
    $\left((1-\alpha)^{\frac{1+\gamma}{\gamma}},0\right)$ along the $v_2$-axis.} 
  \label{fig:vE-plane}
\end{figure}
% %%%%%%%%%%%%%%%
If we make a normalization so that $u_1+u_2 = 1$, then we can write the explicit solution to \eqref{phase_plane} for $E$ as a function of
$v_2$ is   
%along these trajectories 
\begin{equation}\label{traj-soln}
  E(v_2) = v_2 + (\alpha-1)v_2^{\frac{1}{1+\gamma}},
\end{equation}
where $\alpha$ is a parameter that characterizes the trajectories.  We use
$v_2=1$ as a reference value in the $(v_2,E)$-plane, and let $\alpha$ denote
$E(1)$, the value of $E$ when $v_2=1$, which will depend on the initial
conditions.  Specifically, if $E_0$ and $v_2^0$ denote the values of
$E$ and $v_2$ at time $t=0$, then
\begin{equation*}
  \alpha = \frac{E_0-v_2^0}{(v_2^0)^{\frac{1}{1+\gamma}}} + 1.
\end{equation*}
Recall
from \eqref{ODE-u}-\eqref{ODE-ic} that $v_2$ and $E$ are decreasing
whenever $E>0$. The choice $\alpha=1$
corresponds to the trajectory along the line $E=v_2$. When $\alpha>1$,
the trajectory lies above the line $E=v_2$ and approaches the origin as
$t\to \infty$.  When $\alpha<1$,  the trajectory is below the line $E=v_2$,
intersecting the $v_2$-axis at $v_2 = (1-\alpha)^{\frac{1+\gamma}{\gamma}}$.
We summarize these results in the theorem below.

\begin{thm}
\label{ODE-ss-lem}
     Define $E(t)=u_1v_2-u_2v_1$.  If $-v_1(0)<E(0)<v_2(0)$, then the ODE
     system \eqref{ODE-u}-\eqref{ODE-ic} approaches a steady state
     such that
     \begin{equation}
       \frac{u_1}{v_1} = \frac{u_2}{v_2}. \label{ss-condition}
     \end{equation}
      If $v_2(0)\leq E(0)$, then $u_2,v_2\to 0$, whereas if $v_1(0)\leq
      -E(0)$, then $u_1,v_1\to 0$.
\end{thm}

\section{Derivation of the fitness gradient flux PDE}
\label{sec:continuum-limit}
\newcommand{\dt}{\delta t}
\newcommand{\dx}{\delta x}

In this section we derive the fitness gradient flux PDE
\eqref{eq:introPDE} in two dimensions; this system was first described in
\cite{deforest2013}.  Our derivation is similar to continuum limit
arguments for biased
random walks that appear in \cite{Okubo1986, Shigesada1979}.  Biased
random walks in theoretical populations are discussed in greater
detail in  \cite{Skellam1951}.

Let $\{x_{ij}\}$ denote a $J_x\times J_y$ uniform grid with uniform meshsize $\dx$.  At time $t$, each grid point $x_{ij}$ has populations $u(x_{ij},t)$ and
$v(x_{ij},t)$. Our model is based on the following assumption: the movement of each population on this grid is governed by transition probabilities, which are proportional to local differences in fitness, and defined in the following.

\begin{defn}
  Given two grid points $a$ and $b$, and fixed timestep $\dt$, we define the transition
  probability $p(a,b;t)$ to be the probability that an individual
  from population $u$ moves from $a$ to $b$ in
  the time interval $(t,t+\dt)$. We define an analogous transition
  probability $q(a,b;t)$ for the population $v$. 
\end{defn}
Note that the allowed transitions will be made effectively local by restricting points $a$ and $b$ to be nearest neighbors on the grid.
%, however this could be generalized. 
%
We use the following notation conventions throughout this section.
\begin{notation}
  For lattice nodes denoted by $a$, $x_{ij}$, or $x_\alpha$ with 
  $\alpha \in \{ (i,j-1), (i,j+1), (i-1,j), (i+1,j) \}$, let
  \begin{equation*}
    u_a^t := u(a,t), \quad u_{ij}^t = u(x_{ij},t), \quad u_\alpha^t =
    u(x_\alpha, t),
  \end{equation*}
and similarly for the fitness functions $f(u,v)$, $g(u,v)$, let
  \begin{equation*}
    f_{ij}^t = f(u_{ij}^t,v_{ij}^t).
  \end{equation*}
We also define  the following forward-difference and backward
difference operators
\begin{align*}
  D_x^+ u_{ij} = \frac{1}{\dx}\left(u_{i+1,j}-u_{ij}\right),\\
  D_y^+ u_{ij} = \frac{1}{\dx}(u_{i,j+1} -u_{ij}),\\
  D_x^- u_{ij} = \frac{1}{\dx}\left(u_{ij}-u_{i-1,j}\right),\\
  D_y^- u_{ij} = \frac{1}{\dx} \left(u_{ij}-u_{i,j-1}\right).
\end{align*}

\end{notation}

\begin{defn}
  Let $a$ and $b$ be adjacent nodes and let the fitness functions $f$
  and $g$ be bounded continuous functions.  Define the bounds
  \begin{equation*}
    M_1 = \left\{\sup f(\mathbf{x})-\inf f(\mathbf{x})\right\}, \quad
    M_2 = \left\{\sup
      g(\mathbf{x})-\inf g(\mathbf{x})\right\},
  \end{equation*}
where the sup and inf are taken over $\mathbf{x}\in \R^+\times \R^+$
(the domain of the fitness functions).
We define the transition probabilities $p(a,b;t)$ and
  $q(a,b;t)$ to depend on the fitness differences as
  \begin{align}
    \label{transition-probabilities}
      p(a,b;t) &=
      \begin{cases}
        \frac{1}{4 M_1} (f_b^t-f_a^t), & \text{if } f_b^t\geq f_a^t,\\
         0, & \text{if } f_b^t <f_a^t,
      \end{cases}
\\
  q(a,b;t) & =
  \begin{cases}
    \frac{1}{4M_2} (g_b^t-g_a^t),   & \text{if } g_b^t\geq g_a^t, \\
 0, & \text{if }  g_b^t< g_a^t.
  \end{cases}
\notag
  \end{align}

\end{defn}

Note that in this formulation, at most one of $p(a,b;t)$ or  $p(b,a;t)$ can
be nonzero, representing the fact that an individual has nonzero probability of moving to an adjacent node if and only if the fitness is strictly higher at that
node. Thus, populations travel to adjacent points by moving in the direction of increasing fitness, as in the two-node model of Section \ref{sec:ode}.

The scaling constants, $\frac{1}{4M_1}$ and $\frac{1}{4M_2}$, ensure that for any node $x_{ij}$,
\begin{equation*}
\sum_{\alpha}  p(x_{ij},x_\alpha;t)\leq 1, \quad \text{ and }  \quad \sum_{\alpha}
q(x_{ij},x_\alpha;t)\leq 1,
\end{equation*}
 where $\alpha$ ranges over the set of nodes adjacent to $x_{ij}$.

 We now derive the PDE for the population density $u$; the argument for 
 $v$ is entirely similar. Consider  
 $u$ at the node $x_{ij}$ and at time $t+\dt$:  
\begin{align*}
  u(x_{ij},t+\dt)  &= u_{ij}^t +\frac{1}{4M_1} \sum_\alpha
 u_\alpha^t p(x_\alpha,x_{ij};t) -
  \frac{1}{4M_1}\sum_\alpha u_{ij}^t
  p(x_{ij},x_\alpha;t)\\ 
 & = u_{ij}^{t+\dt}.
\end{align*}
With respect to either coordinate direction,
the fitness function $f$ may be increasing, decreasing, or achieve a
local extremum at $x_{ij}$. We show the case where the fitness
function $f$ is increasing with respect to both coordinate directions.
\begin{equation*}
  f_{i+1,j}^t \geq f_{ij}^t \geq f_{i-1,j}^t,\quad f_{i,j+1}^t\geq
  f_{ij}^t \geq f_{i,j-1}^t.
\end{equation*}
\begin{align*}
  u_{ij}^{t+\dt} -u_{ij}^t & = \frac{1}{4M_1}
  u_{i-1,j}^t\left(f_{ij}^t-f_{i-1,j}^t\right) -\frac{1}{4M_1} u_{ij}^t
  \left(f_{i+1,j}^t-f_{ij}^t\right)\\
 &\quad +\frac{1}{4M_1} u_{i,j-1}^t \left(f_{ij}^t-f_{i,j-1}^t\right) -
 \frac{1}{4M_1}u_{ij}\left(f_{i,j+1}^t-f_{ij}^t\right),\\
 \frac{u_{ij}^{t+\dt}-u_{ij}^t}{\dt} & =
 \left(\frac{\dx^2}{4M_1\dt}\right)\left(\frac{u_{i-1,j}^t}{\dx}\left(\frac{f_{ij}^t-f_{i-1,j}^t}{\dx}\right)
   -
   \frac{u_{ij}^t}{\dx}\left(\frac{f_{i+1,j}^t-f_{ij}^t}{\dx}\right)\right)\\
 & \quad + \left(\frac{
     \dx^2}{4M_1\dt}\right)\left(\frac{u_{i,j-1}^t}{\dx}\left(\frac{f_{ij}^t-f_{i,j-1}^t}{\dx}\right)-\frac{u_{ij}^t}{\dx}\left(\frac{f_{i,j+1}^t-f_{ij}^t}{\dx}\right)\right)\\
& = \frac{ \dx^2}{4M_1\dt}\left(\frac{u_{i-1,j}^t}{\dx}D_x^+
  f_{i-1,j}^t-\frac{u_{ij}}{\dx}D_x^+ f_{ij}^t\right)\\
 &\quad +
 \frac{\dx^2}{4M_1\dt}\left(\frac{u_{i,j-1}^t}{\dx}D_y^+
     f_{i,j-1}^t-\frac{u_{ij}}{\dx}D_y^+ f_{ij}\right)\\
 & = \frac{\dx^2}{4M_1\dt}\left[-D_x^-\left(u_{ij}^t D_x^+
     f_{ij}^t\right) - D_y^- \left(u_{ij}^t D_y^+ f_{ij}^t\right)\right].
\end{align*}
Since, by assumption, $-D_x^+ f_{ij}<0$ and
$-D_y^+f_{ij}<0$, we notice that the backward-difference operator in the final line is equivalent to a
first-order upwinding scheme \cite{Lui2011}.  Consideration of the other
cases bears this out.  Therefore, by taking a limit as $\dt\to 0$ and
$\dx\to 0$ in such a way that
\begin{equation*}
  \lim_{\dt\to 0 \atop \dx\to 0} \frac{\dx^2}{\dt} = 1,
\end{equation*}
we arrive at the fitness gradient equation in \eqref{eq:introPDE}, and
given below in \eqref{eq:reduced-form-PDE}, where $\beta_1 =
\frac{1}{4M_1}$ and $\beta_2 = \frac{1}{4M_2}$.

%%% Local Variables: 
%%% mode: latex
%%% TeX-master: "GameTheoreticAggregationMain"
%%% End: 

%%%%%%%%%%%%%
\section{Analysis of the Fitness Gradient Flux System}
\label{sec:analysis}

In this section we analyze the system
\begin{equation}
  \label{eq:reduced-form-PDE}
  \begin{cases}
      \partial_t u = -
  \nabla\cdot\left(u\nabla f\right) , & \text{in } \Omega \times (0,T),\\
  \partial_t v  = -(1+\gamma)
  \nabla\cdot\left(v\nabla f\right) & \text{in } \Omega \times (0,T),\\
  \nu\cdot(u\nabla f) = 0, \quad \nu\cdot(v\nabla f) = 0,&  \text{on } \partial\Omega \times (0,T),\\
  u(\mathbf{x},0) = u_0(\mathbf{x})>0, & \text{in } \Omega,\\
  v(\mathbf{x},0) = v_0(\mathbf{x})>0, & \text{in } \Omega.
  \end{cases}
\end{equation}

\begin{remark}
  Under our assumptions, the populations $u$ and $v$ experience the same fitness
  gradient $\nabla f$, but the population $v$ has a higher
  sensitivity to the gradient than does $u$, since $\gamma > 0$.  The game dynamics lead the population  $u$ to aggregate, and to
  flee regions where the density of $v$ is high, while the population
  $v$ pursues $u$.  Due to the $v$ population's higher sensitivity, it is
  it acts to inhibit $u$'s aggregation. If $\gamma < 0$
  however, then \eqref{eq:reduced-form-PDE} is ill-posed.
   \end{remark}

As discussed in Section~\ref{sec:recast}, we may write
\eqref{eq:reduced-form-PDE} as
% It is sometimes convenient to express \eqref{quasilinear-form2} in the
% more compact form
\begin{equation}\label{vector-form}
  \partial_t \mathbf{w} = \nabla \cdot \left(B(\mathbf{w}) \nabla \mathbf{w}\right),
\end{equation}
where $\mathbf{w} = (u,v)$ and $\nabla \mathbf{w} = (\nabla u, \nabla
v)$ and
  \begin{equation}
    B(u,v) = (b_{ij}(u,v)) =  \frac{1}{(u+v)^2}
    \begin{bmatrix}
      -uv & u^2\\
      -(1+\gamma)v^2 & (1+\gamma) uv
    \end{bmatrix}.
    \label{bmatrix}
  \end{equation}

\subsection{Steady State Solutions}
\label{sec:ss-solns}

As was previously observed in~\cite{deforest2013}, strictly positive steady state solutions are characterized by the simple condition
$u=cv$.
When $u$ and $v$ are smooth positive functions on
$\overline{\Omega}$, then 
%$c = \frac{\norm{u}_{L^1(\Omega)}}{\norm{v}_{L^1(\Omega)}}$.  
we have the following:

\begin{thm}
  Let $u$ and $v$ be strictly positive functions in
  $C^1\left(\overline{\Omega}\right)$.  Then $(u,v)$ is a
  steady state solution of \eqref{eq:reduced-form-PDE} if and only if $u
  = cv$, where
  \begin{equation*}
c = \frac{\norm{u}_{L^1(\Omega)}}{\norm{v}_{L^1(\Omega)}}.
\end{equation*}
\end{thm}

\begin{proof}  
Given $u = cv$, then using $\ln u = \ln (cv)$ one has
\begin{equation*}
  \frac{\nabla u}{u} = \frac{\nabla v}{v},
\end{equation*}
or
\begin{equation*}
  v\nabla u -u\nabla v = 0.
\end{equation*}
From \eqref{eq:fitness-gradient} this implies $\nabla f=\mathbf{0}$  and thus
$\partial_t u = \partial_t v = 0$.  Note that this also implies
that the fitness function $f$ is constant.  

Conversely,  if $(u,v)$ is a steady state solution, then
\begin{align*}
  \nabla\cdot (u\nabla f) & =0,\\
  (1+\gamma)\nabla\cdot (v\nabla f) & = 0.\\
\end{align*}
This implies
\begin{align*}
  -\nabla u \cdot\nabla f & = u \Delta f,\\
  -\nabla v \cdot\nabla f & = v\Delta f.\\
\end{align*}
Consequently,
\begin{equation}\label{steady-state-proof}
  -\frac{\nabla u}{u}\cdot\nabla f  = \Delta f = -\frac{\nabla v}{v}
  \cdot \nabla f,
\end{equation}
and 
\begin{equation}
\label{grad-u-u-proof}
  \left(v\nabla u - u\nabla v\right)\cdot\nabla f = 0.
\end{equation}
Recalling the value of $\nabla f$ from \eqref{eq:fitness-gradient}, we have
\begin{equation*}
  \frac{1}{(u+v)^2}\abs{v\nabla u - u\nabla v}^2 = 0. 
\end{equation*}
Since by our assumptions $\frac{1}{(u+v)^2}>0$, we
conclude
\begin{equation*}
  v\nabla u -u\nabla v = 0.
\end{equation*}
Equivalently,
\begin{equation*}
  \nabla \ln u = \nabla \ln v,
\end{equation*}
which implies $\ln u = \ln (cv)$ and $u=cv$.
Since
\begin{equation*}
  \int_\Omega u dx = \int_\Omega cv dx,
\end{equation*}
with $u,v>0$, it is easy to see that $c$ will be the ratio of the $L^1$ norms. 

\end{proof}

\bigskip

\subsection{Weak Steady State Solutions}
If $\Omega$ is an interval, we can define continuous weak steady state solutions
in $H^1(\Omega)$.  As shown above, when $u=cv$, then we have $\partial_xf=0$; equivalently $f$ constant.
This is a local condition; it possible that $f$ is only
piece-wise constant. If  $u$ and $v$ are to be continuous, we must have
$u=v=0$ at points of discontinuity of  $f$.

For example, suppose that $\Omega$ is partitioned
into two disjoint intervals: $\Omega = (x_0,x_1]\cup (x_1,x_2) =
I_1\cup I_2$ and  let $u = c_1v$ on $I_1$ and $u=c_2v$
on $I_2$.  If each $u$ and $v$ are  to be continuous we must have
$u=v=0$ at the adjoining endpoint $x_1$.  In this case, we have
$\partial_x f=0$ everywhere except at $x_1$ (where $\partial_x f$ is
not defined).  The resulting $(u,v)$ is a weak steady state solution of \eqref{eq:reduced-form-PDE}.

\begin{figure}[h]
  \centering\label{fig:uv1}
  \begin{tikzpicture}
    \begin{axis}[xmin=0,xmax=2,
      ymin=0,ymax=4]
      \addplot[ultra thick,domain=0:1] {-x^2+1};
      \addplot[dashed, ultra thick, domain=0:1] {1.5*(-x^2+1)};
      \addplot[ultra thick,domain=1:2] {-2*((x-2)^2-1)};
      \addplot[dashed, ultra thick, domain=1:2] {-2.5*((x-2)^2-1)};
    \end{axis}
  \end{tikzpicture}
  \caption{$u$ and $v$, $v = 1.5u$ on $(0,1)$ and $v=1.25u$ on (1,2)}
\end{figure}
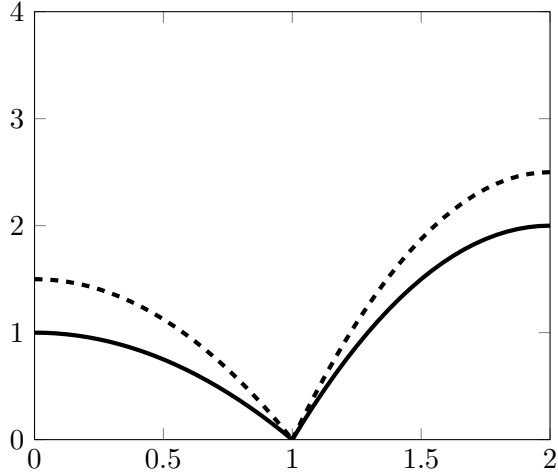

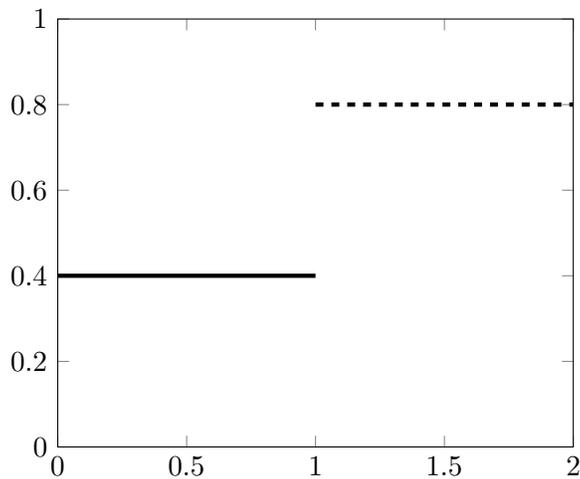
\begin{figure}[h]
  \centering\label{fig:piecewise-fitness}
  \begin{tikzpicture}
    \begin{axis}[xmin=0,xmax=2,
      ymin=0,ymax=1]
      \addplot[ultra thick,domain=0:1] {(1/2.5)+0*x};
      \addplot[dashed, ultra thick, domain=1:2] {2/2.5+0*x};
      % \addplot[ultra thick,domain=1:2] {-2*((x-2)^2-1)};
      % \addplot[dashed, ultra thick, domain=1:2] {-2.5*((x-2)^2-1)};
    \end{axis}
  \end{tikzpicture}
  \caption{Piecewise-constant $f$}
\end{figure}

If we partition $\Omega$ into a set of disjoint
intervals, then $f$ may have a different constant value on each
interval, with $u=v=0$ at the adjoining
endpoints. An example is shown in Figures~\ref{fig:uv1} and \ref{fig:piecewise-fitness}.

\begin{defn}\label{defn:weak-soln-two-populations}
   For functions $u$ and $v$ in $H^1(\Omega)$, $(u,v)$ is a weak
   steady state solution of \eqref{eq:reduced-form-PDE} if, for every
   pair of smooth test functions $\phi,\psi\in C_0^\infty(\Omega)$,
   \begin{equation}
     \label{weak-solution}
     \begin{aligned}
       \int_\Omega \phi_x \cdot( uf_x) dx = 0,\\
       \int_\Omega  \psi_x \cdot (vf_x) dx = 0.
     \end{aligned}
   \end{equation}

\end{defn}

\begin{thm}\label{thm:weak-solution}
  Let $\Omega$ be a bounded open interval in $\R$.  Let $v\in
  H^1(\Omega)$, with $v\geq 0$ and such that $v=0$ at no more than a finite
  number of points $x_k\in\Omega$.  This set of zeros partitions
  $\Omega$ into a finite collection of disjoint intervals $I_k$. 
   
  Construct a function $u$ as follows. For
  each interval $I_k$, let $c_{k}$ be a
  nonnegative  constant and take $u_=c_{k}v$ on $I_k$.  Then $(u,v)$
  is a weak steady state solution of \eqref{eq:reduced-form-PDE}.
\end{thm}

\begin{proof}
  By construction for each $k$, we have $\nabla f =0$ on each
  interval $I_k$.  Since $u$ and $v$ are each in $H^1(\Omega)$, we can
  take $u$ and $v$ to be
  absolutely continuous. Therefore $u(x_k)=v(x_k)=0$.

Take $\phi\in C_0^\infty(\Omega)$ and consider
\begin{equation*}
  \int_\Omega \phi \nabla \cdot (u\nabla f)dx = -\int_\Omega
  \nabla \phi \cdot (u \nabla f) dx
\end{equation*}
Although $u\in H^1(\Omega)$, the function $f$ under our
assumptions is piecewise constant and consequently is not  in $H^1(\Omega)$.  Suppose that $f$ has one point of
discontinuity at  $x_1\in\Omega = (x_0,x_2)$, as in \ref{fig:piecewise-fitness} and
consider a small open interval around this point $B(x_1, \epsilon)$. Then  
\begin{equation*}
  \begin{aligned}
    -\int_\Omega \nabla \phi \cdot(u\nabla f) dx &=
    -\int_{x_0}^{x_1-\epsilon} \nabla \phi \cdot(u \nabla f) dx -
    \int_{x_1+\epsilon}^{x_2} \nabla \phi \cdot (u \nabla f) dx \\
    &\quad
    -\int_{x_1-\epsilon}^{x_1+\epsilon} \nabla \phi\cdot (u \nabla f)
    dx.\\
  \end{aligned}
\end{equation*}
By assumption,  $\nabla f = 0$ on $(x_0,x_1-\epsilon)$ and
$(x_1+\epsilon,x_2)$. Thus we have 
\begin{equation*}
  \begin{aligned}
  -\int_{x_1-\epsilon}^{x_1+\epsilon} \nabla \phi\cdot (u \nabla f)
  dx & =-\int \nabla f\cdot (u\nabla \phi)dx \\
  & = \int_{x_1-\epsilon}^{x_1+\epsilon} \nabla \cdot (u \nabla
    \phi)  f dx -
    (u\nabla\phi f)(x_1+\epsilon) + (u\nabla\phi f)(x_1-\epsilon)\\
  & = \int_{x_1-\epsilon}^{x_1+\epsilon} u \Delta \phi f + \nabla
  u \cdot\nabla \phi f dx -
    (u\nabla\phi f)(x_1+\epsilon) + (u\nabla\phi f)(x_1-\epsilon)
  \end{aligned}
\end{equation*}
Note that $f(u)$ is not defined when $u=v =0$, but $f$ is
bounded as $(u,v)\to (0,0)$.  Since $u$ and $v$ are in $\in H_1(\Omega)$, we have, for some
constant $C$,
\begin{equation*}
  \begin{aligned}
   \abs{ \int_\Omega \nabla \phi (u\nabla f) dx} & \leq \epsilon C
    \norm{u}_{H^1(\Omega)}\norm{\phi}_{H^1(\Omega)} \to 0 \text{ as
    } \epsilon \to 0.
  \end{aligned}
\end{equation*}
Similarly, we have  $\int_\Omega \phi \nabla \cdot (v\nabla f)dx = 0$
so that $(u,v)$ is a weak solution of \eqref{eq:reduced-form-PDE}. 
\end{proof}

\begin{remark}
  Although the piecewise constant function $f$ does not have a weak
  derivative, its distributional derivative is a delta function (or a finite set
  of delta functions in the general case). Integrating the function
  $u$ against $\nabla
  f$ thus gives us the value $u(x_1)$, which by our assumptions is
  zero. Thus we see (again) that it is essential that the function $u=0$ at
  each point $x_k$ where the fitness $f$ is discontinuous.
\end{remark}

\begin{remark}
  For the model in \eqref{eq:FFPDE1} below, describing two
  populations, numerical simulations have shown that for some initial
  conditions, the system evolves to such weak steady state solutions.
  A  'pinching off' occurs, where each population reaches zero at a
  point in $\Omega$. The populations then redistribute themselves on
  the remaining  subintervals, until reaching a configuration where
  $\nabla f=0$ on each subinterval.
\end{remark}

\smallskip

\begin{remark}
  The consideration of weak steady state solutions reveals two 
  short-comings in our model. First, there is no law of motion for the
  population $u$ in the absence of $v$, or vice versa (since the fitness is constant in that case).  If $u_0=0$ on
  some subinterval $I\subset \Omega$, then $(u_0,v_0)$ will be a weak
  steady state solution to \eqref{eq:reduced-form-PDE}, provided $u_0$ and $v_0$ 
  are in $H^1(\Omega)$, satisfy the Neumann boundary condition, and 
  \begin{equation}
  \label{eq:pathological-example}
    u_0 = cv_0  \text{ in } \Omega\setminus I.\\
\end{equation}
That is, $v_0$ can be arbitrarily
  chosen on the subinterval $I$ where $u_0\equiv0$.

Second, our fitness derives from an evolutionary game, which
is inherently a mean-field  model.  The evolutionary game approach
assumes large well-mixed populations, but these assumptions break down
when $u+v\ll 1$. An improved model would require multiple scales, where
the mean-field approach dominates when $u$ and $v$ are large, while
dynamics for individual interactions are brought into play when $u$ and
$v$ are near zero.

In future work we may consider alterations to these models that
address these shortcomings.
\end{remark}

%%%%%%%%%%%%%%%%%%%%%%%%%%%%%%%%%%%%%%%%%5
% REMARK
% Remark on how evolutionary game relative fitness is not reasonable
% when u and v approach zero.  (limitations of model) The evolution game idea relies on large
% well mixed populations.  When populations are small, we may need to
% consider multi-scale models, coupling this model with a model that
% describes behavior of small populations.  Our model also is not
% appropriate when one population is absent.  What is the rule of
% motion governing the population of u in the absence of v or vice
% versa?  
%  Include paradoxical steady state solutions,   u = cv where c =0 on
%  some intervals.  u can have any form.
%
%\begin{remark}
%  
%\end{remark}
% End REMARK
%%%%%%%%%%%%%%%%%%%%%%%%%%%%%%%%%%%

%%%%%%%%%%%%%%%%%%%%%%%%%%%%%%%%%%%%%%%%%%%%%%%%%%%%%%%%%%%%%%%%%%%%%%%

\subsection{Linearization around a steady state}
\label{sec:line-around-steady-1}

We next study solutions for a linearization of the fitness-flux PDE in
the case of two populations, and where $\Omega$ is an
interval. Solutions are of the form
$$
\mathbf{w}(x,t) = (u(x,t)-u_0(x),v(x,t)-v_0(x)),
$$ 
where $(u_0,v_0)$ is a smooth (strictly positive),  steady state solution to \eqref{eq:FFPDE1}.
We show that this steady state is neutrally. While perturbations from the steady state remain bounded, they do not decay but tend toward a new steady state near $(u_0,v_0)$ in the $L^2$ sense.

We choose $\Omega$ to be the interval $(0,2\pi)$; $u_0$ and $v_0$ are strictly positive, and $u_0 = cv_0$ with $c= \norm{u_0} / \norm{v_0}$ as before. 
We consider the linearization 
\begin{equation*}
  \begin{aligned}
    \partial_t w_1 &= \nabla \cdot \left[\partial_uP(u_0,v_0)w_1
      + \partial_v P(u_0,v_0)w_2\right]\\
      \partial_t w_2 & = \nabla \cdot \left[\partial_u
        Q(u_0,v_0)w_1+ \partial_v(u_0,v_0)w_2\right],
  \end{aligned}
\end{equation*}
Where
\begin{equation*}
  \begin{aligned}
    P(u,v) &= \nabla \cdot [b_{11}(u,v)\nabla u + b_{12}(u,v) \nabla
    v]\\
    Q(u,v) & = \nabla \cdot [b_{21}(u,v)\nabla v + b_{22}(u,v)\nabla
    v],
  \end{aligned}
\end{equation*}

\begin{equation*}
  \begin{aligned}
    \partial_u P(u,v) w_1 & = \nabla \cdot
    \left[\frac{uv-v^2}{(u+v)^3}w_1\nabla u - \frac{uv}{(u+v)^2}\nabla
    w_1 + \frac{2uv}{(u+v)^3}w_1\nabla v\right]\\
  \partial_v P(u,v) w_2 & = \nabla \cdot
  \left[\frac{uv-u^2}{(u+v)^3}w_2\nabla u + \frac{u^2}{(u+v)^2}\nabla
    w_2 - \frac{2u^2}{(u+v)^3}w_2 \nabla v\right]\\
  \partial_u Q(u,v)w_1 & = (1+\gamma)\nabla \cdot
  \left[\frac{2v^2}{(u+v)^3}w_1\nabla u - \frac{v^2}{(u+v)^2}\nabla
    w_1 + \frac{v^2-uv}{(u+v)^3}w_1\nabla v\right]\\
  \partial_v Q(u,v) w_2 & = (1+\gamma)\nabla \cdot
  \left[\frac{-2uv}{(u+v)^3}w_2\nabla u +\frac{uv}{(u+v)^2}\nabla w_2
    +\frac{u^2-uv}{(u+v)^3}w_2\nabla v \right]
  \end{aligned}
\end{equation*}
Evaluating the above at the steady state solution $(cv_0,v_0)$, we
arrive at the following linearized PDE:
\begin{equation}
  \label{eq:linear-ffpde}
  \begin{cases}
    \partial_t w_1 &= \frac{1}{(c+1)^2}\nabla \cdot
    \left[(c w_1-c^2w_2)\frac{\nabla v_0}{v_0} -c\nabla w_1 +c^2 \nabla
      w_2\right]\\
    \partial_t w_2 &= \frac{(1+\gamma)}{(c+1)^2}\nabla
    \cdot \left[(w_1-cw_2)\frac{\nabla v_0}{v_0}  -\nabla w_1+c\nabla w_2 \right]
\end{cases}
\end{equation}
which can be written  as 
\begin{equation*}
    \begin{aligned}
      \partial_t \mathbf{w} = K\nabla \cdot \left(B\nabla \mathbf{w} -
        \frac{\nabla v_0}{v_0}B \mathbf{w} \right)
    \end{aligned}
  \end{equation*}
for
\begin{equation*}
K = \frac{1}{(c+1)^2},\quad  \mathbf{w} = (w_1,w_2),\quad \text{ and } \quad B =
  \begin{pmatrix}
    - c & c^2\\ -(1+\gamma) & (1+\gamma)c
  \end{pmatrix}
\end{equation*}

The eigenvalues of $B$ are $\lambda_0 = 0$ and $\lambda_1
= \gamma K c>0$, with corresponding eigenvectors
\begin{equation*}
  \mathbf{e}_0 =
  \begin{bmatrix}
    c\\1
  \end{bmatrix}
, \quad \mathbf{e}_1 =
\begin{bmatrix}
  \frac{c}{1+\gamma}\\1
\end{bmatrix}.
\end{equation*}

Using the eigenvectors
given above, we can decompose $\mathbf{w}$ as
\begin{equation*}
  \begin{aligned}
    \mathbf{w}(x,t) = c_0(x,t)\mathbf{e}_0  + c_1(x,t)\mathbf{e}_1, 
  \end{aligned}
\end{equation*}
where $c_0(x,t)\mathbf{e}_0$ is in the eigenspace associated with
$\lambda_0$ and, hence remains constant in time, while $c_1(x,t)$ will
evolve according to the linear PDE shown below (see equation~\eqref{eq:decoupled-pde}).
Solving this system gives
\begin{equation}
  \label{eq:1}
  \begin{aligned}
    c_0(x,t) &= \frac{1+\gamma}{c\gamma}w_1(x,t)-\frac{1}{\gamma}w_2(x,t), 
    \\
    c_1 (x,t) & = - \frac{1+\gamma}{c\gamma}w_1(x,t)+\frac{1+\gamma}{\gamma}w_2(x,t). 
  \end{aligned}
\end{equation}
Given an initial condition $\mathbf{w_0}(x,0) = (w_1,w_2)$, the
function $c_0(x,t)\mathbf{e}_0 = c_0(x,0)\mathbf{e}_0$ is constant in time. Writing this
as $\mathbf{y}_0(x) = (y_0^1,y_0^2)$, we have $y_0^1 = cy_0^2$, as we expect.

Writing $\textbf{y}_1(x,t) = c_1(x,t)\mathbf{e}_1$, we see that $\partial_t
y_1^1 = \frac{c}{1+\gamma}\partial_t y_1^2$.  Thus, we can reduce the
problem to the single linear partial differential equation,

\begin{equation}
  \label{eq:decoupled-pde}
    \partial_t w = \alpha \Delta w - a(x)\nabla w - b(x)w, \quad \text{where } \quad \alpha =
    K\gamma c>0,
\end{equation}
and
\begin{equation*}
a(x) = \alpha\frac{\nabla v_0}{v_0}, \quad b(x)
= \alpha\left(\frac{\Delta v_0}{v_0} - \frac{\abs{\nabla v_0}^2}{v_0^2}\right).
\end{equation*}
A perturbation attains a new steady state.  We illustrate an example
in Figure~\ref{fig:linearization-sim2} in the next section.

% We illustrate an example solution $y_1(x,t)$, to \eqref{eq:decoupled-pde} in Figure~\ref{fig:linearization_ex1}. The
% steady state we have linearized around is $(2v_0,v_0)$, with
% $\gamma=0.5$ and where $v_0(x)
% = \cos(2x)e^{-(\pi-x)^2}+2$ so that $c = 2$ and $\alpha = \frac{1}{9}.$ 
% \begin{figure}[h]
%   \centering
%    \setlength{\figwidth}{.4\textwidth}
%    \includegraphics[width=.8\textwidth]{linearization-soln-ex1}

%   %\includegraphics[width=\textwidth]{solution_example2}
%   \caption{An initial condition $y_1(x,0) =0.5\cos(1.5
%     x)+0.3\cos(2.5x)+1$ (in blue) and the resulting steady state
%     (orange), for \eqref{eq:decoupled-pde} where $\alpha=1/9$, $c=2$, and $\gamma = 0.5$.}
% \label{fig:linearization_ex1}
% \end{figure}

We can also see the instability in the linearization by investigating
a  dispersion relation.  We assume the solution
takes the form $w(x,t) = e^{i(\mathbf{k}\cdot \mathbf{x}+\omega t)}$, with $\omega\in \mathbb{C}$ and
$\mathbf{k}$ and $\mathbf{x}$ in 
$\R^n$. Plugging this into
\eqref{eq:decoupled-pde} gives the dispersion relation
\begin{equation*}
  \frac{i \omega}{\alpha} = -\abs{\mathbf{k}}^2 -i\mathbf{k}\cdot a(x) - b(x).
\end{equation*}
Thus, the real part of $i\omega = -\alpha \abs{\mathbf{k}}^2-b(x)$.
Since $b(x)$ is not positive in general, the modes for for small wave-numbers,
may grow on some parts of the domain $\Omega$, while
for sufficiently large $\abs{k}$,  the associated modes will
decay.

\section{Discussion and Numerical Examples}
\label{sec:numerics}

To illustrate and provide insight into the results presented above, we next discuss several numerical examples in 1D. We discuss the numerical methods in Section \ref{methods}.   In Sections \ref{solution-example} - \ref{perturbation-example}, we examine the transient and perturbation dynamics of steady states. In \ref{weak-soln-example} we show several examples evolving to a
weak steady state, with piecewise constant fitness. We conclude by
demonstrating cross-diffusive instabilities and the onset of pattern formation
produced when fitness gradient flux is included in a Lotka--Volterra type
population model (Section \ref{sec:lv}).

\subsection{Numerical Methods}
\label{methods}
We use an implicit numerical scheme by Newton iteration with a no-flux
boundary condition; the discretization uses a first order upwinding scheme, 
%where the upwind direction depends on the sign of $\nabla f$.  
%Upwinding is
necessary for simulating examples that evolve toward a weak steady
state solution.

For the one-dimensional case, the PDE system \eqref{eq:reduced-form-PDE}
can be written as
\begin{align*}
  u_t &= -u_xf_x - uf_{xx},\\
  v_t &= -(1+\gamma)v_xf_x - (1+\gamma)vf_{xx}.
\end{align*}
We describe the first-order upwinding discretization for $u_t$. The sign of $f_x$ determines whether we use
a forward or backward difference in the discretization of $u_x$: we
use a backward difference when $f_x>0$ and a forward
difference when $f_x<0$ \cite{Lui2011}.

Let $f_{x,i}^+$ and $f_{x,i}^-$ denote the forward and backward
difference operators at $x_i$,
\begin{equation*}
  f_{x,i}^+ = \frac{1}{\delta x}\left(f_{i+1}-f_i\right), \quad
  f_{x,i}^- = \frac{1}{\delta x}\left(f_i-f_{i-1}\right), 
\end{equation*}
where $f_i$ denotes $f(x_i)$ (we are supressing the time variable $t$).

The first order central difference for $f_{xx}$ can be computed as
\begin{align*}
  f_{xx,i} & = \frac{1}{\delta x} \left(f_{x,i}^+-f_{x,i}^-\right).
\end{align*}
If we use $f_{x,i}^-$ whenever $f_{x,i}>0$ and
$f_{x,i}^+$ whenever $f_{x,i}<0$, then we have the discretization
\begin{align*}
  \partial_t u_i  =
  \begin{cases}
    -u_{i+1}f_{x,i}^+ + u_i f_{x,i}^-, & f_{x,i}<0,\\
    -u_if_{x,i}^+ + u_{i-1}f_{x,i}^-, & f_{x,i}>0,
  \end{cases}
\end{align*}
which we combine into
\begin{equation*}
  \partial_t u_i = -u_if_{x,i}^+[f_{x,i}^+\geq 0] -
  u_{i+1}f_{x,i}^+[f_{x,i}^+<0] + u_{i-1}f_{x,i}^-[f_{x,i}^-\geq 0] + u_if_{x,i}^-[f_{x,i}^-<0],
\end{equation*}
where
\begin{equation*}
  [g\geq 0] =
  \begin{cases}
    1, & \text{if } g\geq 0,\\
    0, & \text{otherwise}.
  \end{cases}
\end{equation*}
\begin{remark}
  If $f_{x,i}^+$ and $f_{x,i}^-$ differ in sign for some $x_i$, then
  the above discritization treats $f_{x,i}= 0$.
\end{remark}

For the examples shown below, we use a uniform mesh size $\delta x =
0.005$, and $\delta t = 0.001$ on the domain $[0,1]$, with $J=201$
gridpoints.  Refining the mesh and
reducing the time step ($\delta x=0.001,J=1001, \delta t = 10^{-5}$)
does not produce a significant difference in the results.
 
\subsection{Evolution toward steady state solutions}
\label{solution-example}

Beginning from arbitrary but smooth initial conditions, a typical solution exhibits two distinct phases in its dynamics. First, the populations quickly reach a configuration where local extrema of $u$ and $v$ are aligned with one another, as well as  with
the local extrema of the fitness function $f$.
Once aligned, the local maxima of $u$ and $v$ grow while their local
minima decrease, but at a decreasing rate as $\nabla f\to 0$, and the 
solution approaches a steady state.

%%%%%%%%%%%%%%%%%%%%%%%%%%%%%%
\begin{figure}[t!]%[h]
  \centering
  \setlength{\figwidth}{.43\textwidth}
  \input{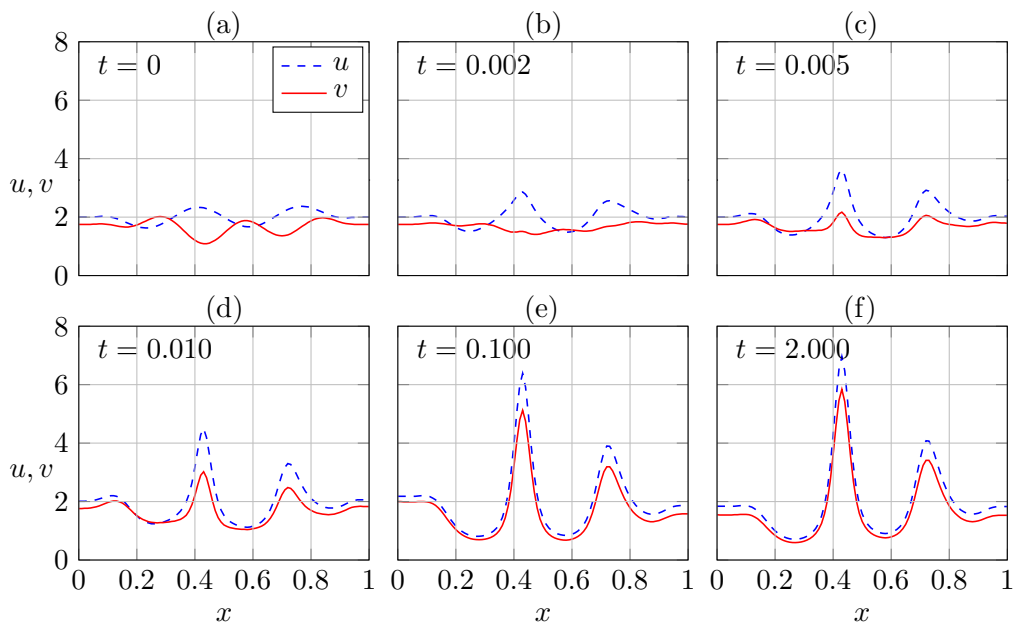}

  \caption{Evolution of a solution to \protect \eqref{eq:reduced-form-PDE} 
  toward steady state. The extrema for $u$ and $v$ are 
  almost perfectly aligned by the tenth iteration (d), $t=0.010$.  The
  local 
  maxima grow at a
  decreasing rate as the solution approaches a steady state in (e), where 
  $u=cv$. At steady state (f), the fitness function $f(u,v)$ is constant 
  throughout $\Omega$. Solutions shown are for $\gamma = 0.5$.}
\label{solution1}
\end{figure}
%%%%%%%%%%%%%%%%%%%%%%%%%%%%%%

\begin{figure}[t]
  \centering
  \setlength{\figwidth}{.7\textwidth}
  % This file was created by matplotlib v0.1.0.
% Copyright (c) 2010--2014, Nico Schlömer <nico.schloemer@gmail.com>
% All rights reserved.
% 
% The lastest updates can be retrieved from
% 
% https://github.com/nschloe/matplotlib2tikz
% 
% where you can also submit bug reports and leavecomments.
% 
\begin{tikzpicture}

\begin{axis}[
xlabel={$x$},
ylabel={$u,v$},
ylabel style={rotate=-90},
axis y line*=left,
xmin=0, xmax=1,
ymin=0, ymax=5,
axis on top,
width=\figwidth,
xmajorgrids,
ymajorgrids,
legend entries={{$u$},{$v$},{$f$}}
]
\addplot [semithick, blue, dashed]
coordinates {
(0,2.01)
(0.005,2.01)
(0.01,2.01)
(0.015,2.01)
(0.02,2.01)
(0.025,2.02)
(0.03,2.02)
(0.035,2.02)
(0.04,2.03)
(0.045,2.03)
(0.05,2.04)
(0.055,2.05)
(0.06,2.06)
(0.065,2.07)
(0.07,2.09)
(0.075,2.1)
(0.08,2.12)
(0.085,2.14)
(0.09,2.15)
(0.095,2.17)
(0.1,2.18)
(0.105,2.19)
(0.11,2.2)
(0.115,2.2)
(0.12,2.2)
(0.125,2.19)
(0.13,2.18)
(0.135,2.16)
(0.14,2.13)
(0.145,2.09)
(0.15,2.05)
(0.155,2.01)
(0.16,1.96)
(0.165,1.9)
(0.17,1.84)
(0.175,1.78)
(0.18,1.72)
(0.185,1.67)
(0.19,1.61)
(0.195,1.56)
(0.2,1.51)
(0.205,1.46)
(0.21,1.42)
(0.215,1.38)
(0.22,1.35)
(0.225,1.33)
(0.23,1.3)
(0.235,1.28)
(0.24,1.27)
(0.245,1.25)
(0.25,1.24)
(0.255,1.24)
(0.26,1.24)
(0.265,1.24)
(0.27,1.25)
(0.275,1.25)
(0.28,1.26)
(0.285,1.28)
(0.29,1.29)
(0.295,1.31)
(0.3,1.33)
(0.305,1.35)
(0.31,1.37)
(0.315,1.4)
(0.32,1.43)
(0.325,1.46)
(0.33,1.49)
(0.335,1.53)
(0.34,1.57)
(0.345,1.62)
(0.35,1.67)
(0.355,1.73)
(0.36,1.81)
(0.365,1.89)
(0.37,1.99)
(0.375,2.11)
(0.38,2.26)
(0.385,2.44)
(0.39,2.65)
(0.395,2.9)
(0.4,3.18)
(0.405,3.47)
(0.41,3.77)
(0.415,4.06)
(0.42,4.29)
(0.425,4.45)
(0.43,4.5)
(0.435,4.4)
(0.44,4.19)
(0.445,3.92)
(0.45,3.6)
(0.455,3.27)
(0.46,2.94)
(0.465,2.65)
(0.47,2.39)
(0.475,2.17)
(0.48,1.98)
(0.485,1.84)
(0.49,1.71)
(0.495,1.62)
(0.5,1.53)
(0.505,1.47)
(0.51,1.41)
(0.515,1.36)
(0.52,1.32)
(0.525,1.29)
(0.53,1.26)
(0.535,1.23)
(0.54,1.21)
(0.545,1.19)
(0.55,1.17)
(0.555,1.16)
(0.56,1.15)
(0.565,1.14)
(0.57,1.13)
(0.575,1.13)
(0.58,1.12)
(0.585,1.13)
(0.59,1.13)
(0.595,1.14)
(0.6,1.16)
(0.605,1.17)
(0.61,1.19)
(0.615,1.21)
(0.62,1.24)
(0.625,1.27)
(0.63,1.31)
(0.635,1.36)
(0.64,1.41)
(0.645,1.47)
(0.65,1.55)
(0.655,1.63)
(0.66,1.73)
(0.665,1.84)
(0.67,1.97)
(0.675,2.12)
(0.68,2.27)
(0.685,2.44)
(0.69,2.61)
(0.695,2.77)
(0.7,2.93)
(0.705,3.06)
(0.71,3.17)
(0.715,3.25)
(0.72,3.3)
(0.725,3.3)
(0.73,3.27)
(0.735,3.21)
(0.74,3.13)
(0.745,3.03)
(0.75,2.93)
(0.755,2.83)
(0.76,2.72)
(0.765,2.62)
(0.77,2.53)
(0.775,2.44)
(0.78,2.36)
(0.785,2.29)
(0.79,2.22)
(0.795,2.16)
(0.8,2.11)
(0.805,2.06)
(0.81,2.02)
(0.815,1.98)
(0.82,1.94)
(0.825,1.91)
(0.83,1.88)
(0.835,1.86)
(0.84,1.84)
(0.845,1.82)
(0.85,1.8)
(0.855,1.79)
(0.86,1.78)
(0.865,1.78)
(0.87,1.78)
(0.875,1.78)
(0.88,1.78)
(0.885,1.79)
(0.89,1.8)
(0.895,1.81)
(0.9,1.83)
(0.905,1.85)
(0.91,1.87)
(0.915,1.9)
(0.92,1.92)
(0.925,1.95)
(0.93,1.97)
(0.935,1.99)
(0.94,2.01)
(0.945,2.03)
(0.95,2.04)
(0.955,2.05)
(0.96,2.06)
(0.965,2.06)
(0.97,2.06)
(0.975,2.06)
(0.98,2.06)
(0.985,2.06)
(0.99,2.06)
(0.995,2.06)
(1,2.06)

};
\addplot [semithick, red]
coordinates {
(0,1.76)
(0.005,1.76)
(0.01,1.76)
(0.015,1.77)
(0.02,1.77)
(0.025,1.77)
(0.03,1.77)
(0.035,1.78)
(0.04,1.78)
(0.045,1.79)
(0.05,1.8)
(0.055,1.81)
(0.06,1.82)
(0.065,1.84)
(0.07,1.85)
(0.075,1.87)
(0.08,1.89)
(0.085,1.91)
(0.09,1.93)
(0.095,1.95)
(0.1,1.97)
(0.105,1.99)
(0.11,2)
(0.115,2.01)
(0.12,2.02)
(0.125,2.02)
(0.13,2.02)
(0.135,2.01)
(0.14,1.99)
(0.145,1.97)
(0.15,1.95)
(0.155,1.91)
(0.16,1.88)
(0.165,1.83)
(0.17,1.79)
(0.175,1.74)
(0.18,1.7)
(0.185,1.65)
(0.19,1.6)
(0.195,1.56)
(0.2,1.52)
(0.205,1.48)
(0.21,1.45)
(0.215,1.42)
(0.22,1.39)
(0.225,1.37)
(0.23,1.35)
(0.235,1.33)
(0.24,1.31)
(0.245,1.3)
(0.25,1.29)
(0.255,1.28)
(0.26,1.28)
(0.265,1.28)
(0.27,1.27)
(0.275,1.28)
(0.28,1.28)
(0.285,1.28)
(0.29,1.28)
(0.295,1.29)
(0.3,1.29)
(0.305,1.3)
(0.31,1.3)
(0.315,1.31)
(0.32,1.32)
(0.325,1.33)
(0.33,1.33)
(0.335,1.35)
(0.34,1.36)
(0.345,1.37)
(0.35,1.39)
(0.355,1.42)
(0.36,1.45)
(0.365,1.49)
(0.37,1.53)
(0.375,1.6)
(0.38,1.68)
(0.385,1.78)
(0.39,1.9)
(0.395,2.05)
(0.4,2.21)
(0.405,2.4)
(0.41,2.58)
(0.415,2.75)
(0.42,2.9)
(0.425,2.99)
(0.43,3.02)
(0.435,2.96)
(0.44,2.82)
(0.445,2.65)
(0.45,2.45)
(0.455,2.24)
(0.46,2.04)
(0.465,1.85)
(0.47,1.69)
(0.475,1.56)
(0.48,1.45)
(0.485,1.36)
(0.49,1.3)
(0.495,1.24)
(0.5,1.2)
(0.505,1.17)
(0.51,1.15)
(0.515,1.13)
(0.52,1.11)
(0.525,1.1)
(0.53,1.09)
(0.535,1.08)
(0.54,1.07)
(0.545,1.06)
(0.55,1.06)
(0.555,1.05)
(0.56,1.05)
(0.565,1.05)
(0.57,1.05)
(0.575,1.04)
(0.58,1.04)
(0.585,1.05)
(0.59,1.05)
(0.595,1.06)
(0.6,1.07)
(0.605,1.08)
(0.61,1.09)
(0.615,1.1)
(0.62,1.12)
(0.625,1.14)
(0.63,1.16)
(0.635,1.19)
(0.64,1.22)
(0.645,1.26)
(0.65,1.31)
(0.655,1.36)
(0.66,1.43)
(0.665,1.5)
(0.67,1.59)
(0.675,1.69)
(0.68,1.79)
(0.685,1.91)
(0.69,2.02)
(0.695,2.13)
(0.7,2.24)
(0.705,2.33)
(0.71,2.4)
(0.715,2.46)
(0.72,2.48)
(0.725,2.48)
(0.73,2.45)
(0.735,2.41)
(0.74,2.35)
(0.745,2.28)
(0.75,2.21)
(0.755,2.14)
(0.76,2.07)
(0.765,2.01)
(0.77,1.95)
(0.775,1.89)
(0.78,1.84)
(0.785,1.8)
(0.79,1.76)
(0.795,1.73)
(0.8,1.7)
(0.805,1.67)
(0.81,1.65)
(0.815,1.63)
(0.82,1.61)
(0.825,1.6)
(0.83,1.59)
(0.835,1.58)
(0.84,1.57)
(0.845,1.56)
(0.85,1.56)
(0.855,1.55)
(0.86,1.55)
(0.865,1.56)
(0.87,1.56)
(0.875,1.56)
(0.88,1.57)
(0.885,1.58)
(0.89,1.6)
(0.895,1.61)
(0.9,1.63)
(0.905,1.65)
(0.91,1.67)
(0.915,1.69)
(0.92,1.71)
(0.925,1.74)
(0.93,1.76)
(0.935,1.78)
(0.94,1.8)
(0.945,1.81)
(0.95,1.82)
(0.955,1.83)
(0.96,1.83)
(0.965,1.84)
(0.97,1.84)
(0.975,1.83)
(0.98,1.83)
(0.985,1.83)
(0.99,1.83)
(0.995,1.83)
(1,1.83)

};
\addplot [semithick, black, dashdotted]
coordinates {
(0,-1)
(1,-1)

};
\path [draw=black, fill opacity=0] (axis cs:13,5)--(axis cs:13,5);

\path [draw=black, fill opacity=0] (axis cs:1,13)--(axis cs:1,13);

\path [draw=black, fill opacity=0] (axis cs:13,0)--(axis cs:13,0);

\path [draw=black, fill opacity=0] (axis cs:0,13)--(axis cs:0,13);

\node at (axis cs:0.02,4.7)[
  anchor=west,
  text=black,
  rotate=0.0
]{$t=0.010$};
\end{axis}

\begin{axis}[
ylabel={$f$},
axis y line*=right,
yticklabel style={
        /pgf/number format/fixed,
        /pgf/number format/precision=2
},
scaled y ticks=false,
ylabel style= {rotate=-90},
xmin=0, xmax=1,
ymin=-0.05, ymax=0.2,
axis on top,
width=\figwidth,
]
\addplot [semithick, black,dashdotted]
coordinates {
(0,0.0655)
(0.005,0.0654)
(0.01,0.0654)
(0.015,0.0652)
(0.02,0.065)
(0.025,0.0648)
(0.03,0.0645)
(0.035,0.0641)
(0.04,0.0637)
(0.045,0.0631)
(0.05,0.0626)
(0.055,0.0619)
(0.06,0.0611)
(0.065,0.0602)
(0.07,0.0593)
(0.075,0.0582)
(0.08,0.057)
(0.085,0.0556)
(0.09,0.0542)
(0.095,0.0526)
(0.1,0.0509)
(0.105,0.0491)
(0.11,0.0471)
(0.115,0.045)
(0.12,0.0428)
(0.125,0.0404)
(0.13,0.038)
(0.135,0.0354)
(0.14,0.0327)
(0.145,0.0298)
(0.15,0.0269)
(0.155,0.0239)
(0.16,0.0208)
(0.165,0.0176)
(0.17,0.0145)
(0.175,0.0112)
(0.18,0.00803)
(0.185,0.00486)
(0.19,0.00175)
(0.195,-0.00126)
(0.2,-0.00413)
(0.205,-0.00682)
(0.21,-0.0093)
(0.215,-0.0115)
(0.22,-0.0135)
(0.225,-0.0151)
(0.23,-0.0164)
(0.235,-0.0173)
(0.24,-0.0178)
(0.245,-0.0178)
(0.25,-0.0175)
(0.255,-0.0166)
(0.26,-0.0153)
(0.265,-0.0135)
(0.27,-0.0111)
(0.275,-0.00833)
(0.28,-0.00501)
(0.285,-0.00118)
(0.29,0.00315)
(0.295,0.00799)
(0.3,0.0133)
(0.305,0.0191)
(0.31,0.0254)
(0.315,0.0322)
(0.32,0.0394)
(0.325,0.0471)
(0.33,0.0551)
(0.335,0.0636)
(0.34,0.0724)
(0.345,0.0815)
(0.35,0.0909)
(0.355,0.1)
(0.36,0.11)
(0.365,0.12)
(0.37,0.13)
(0.375,0.139)
(0.38,0.148)
(0.385,0.157)
(0.39,0.165)
(0.395,0.172)
(0.4,0.178)
(0.405,0.184)
(0.41,0.188)
(0.415,0.191)
(0.42,0.194)
(0.425,0.195)
(0.43,0.196)
(0.435,0.196)
(0.44,0.195)
(0.445,0.193)
(0.45,0.191)
(0.455,0.187)
(0.46,0.183)
(0.465,0.177)
(0.47,0.171)
(0.475,0.163)
(0.48,0.156)
(0.485,0.147)
(0.49,0.139)
(0.495,0.13)
(0.5,0.121)
(0.505,0.112)
(0.51,0.104)
(0.515,0.0957)
(0.52,0.088)
(0.525,0.0807)
(0.53,0.0738)
(0.535,0.0675)
(0.54,0.0618)
(0.545,0.0566)
(0.55,0.052)
(0.555,0.048)
(0.56,0.0446)
(0.565,0.0419)
(0.57,0.0398)
(0.575,0.0383)
(0.58,0.0375)
(0.585,0.0373)
(0.59,0.0378)
(0.595,0.0388)
(0.6,0.0405)
(0.605,0.0428)
(0.61,0.0456)
(0.615,0.0489)
(0.62,0.0528)
(0.625,0.0571)
(0.63,0.0619)
(0.635,0.067)
(0.64,0.0724)
(0.645,0.0781)
(0.65,0.0839)
(0.655,0.0898)
(0.66,0.0957)
(0.665,0.102)
(0.67,0.107)
(0.675,0.112)
(0.68,0.117)
(0.685,0.122)
(0.69,0.126)
(0.695,0.13)
(0.7,0.133)
(0.705,0.136)
(0.71,0.138)
(0.715,0.14)
(0.72,0.141)
(0.725,0.142)
(0.73,0.143)
(0.735,0.143)
(0.74,0.142)
(0.745,0.141)
(0.75,0.14)
(0.755,0.138)
(0.76,0.136)
(0.765,0.133)
(0.77,0.13)
(0.775,0.127)
(0.78,0.123)
(0.785,0.12)
(0.79,0.116)
(0.795,0.112)
(0.8,0.108)
(0.805,0.104)
(0.81,0.1)
(0.815,0.0964)
(0.82,0.0927)
(0.825,0.0891)
(0.83,0.0856)
(0.835,0.0823)
(0.84,0.0792)
(0.845,0.0763)
(0.85,0.0736)
(0.855,0.0711)
(0.86,0.0688)
(0.865,0.0668)
(0.87,0.0649)
(0.875,0.0633)
(0.88,0.0619)
(0.885,0.0607)
(0.89,0.0597)
(0.895,0.0588)
(0.9,0.0582)
(0.905,0.0576)
(0.91,0.0573)
(0.915,0.057)
(0.92,0.0569)
(0.925,0.0568)
(0.93,0.0569)
(0.935,0.057)
(0.94,0.0571)
(0.945,0.0573)
(0.95,0.0576)
(0.955,0.0578)
(0.96,0.0581)
(0.965,0.0583)
(0.97,0.0585)
(0.975,0.0588)
(0.98,0.0589)
(0.985,0.0591)
(0.99,0.0592)
(0.995,0.0593)
(1,0.0594)

};
\path [draw=black, fill opacity=0] (axis cs:13,0.2)--(axis cs:13,0.2);

\path [draw=black, fill opacity=0] (axis cs:1,13)--(axis cs:1,13);

\path [draw=black, fill opacity=0] (axis cs:13,-0.05)--(axis cs:13,-0.05);

\path [draw=black, fill opacity=0] (axis cs:0,13)--(axis cs:0,13);

\end{axis}

\end{tikzpicture}

%%%%%%%%%%%%%%%%%%%%%%%%%%%%%%%%%%%%%%%%%%%%%%%%%%%%%%%

%%% Local Variables: 
%%% mode: latex
%%% TeX-master: "GameTheoreticDispersalMain"
%%% End: 

  \caption{The plot from  Figure~\ref{solution1}d
    compared with the fitness $f(u,v)$.
 The extrema of $u$, $v$,
    and $f$ are
    nearly aligned.  Notice that $f$ has a different scale than $u$
    and $v$. In particular, $f$ falls below zero over part of the
    domain.}
  \label{fig:aligned-extrema}
\end{figure}
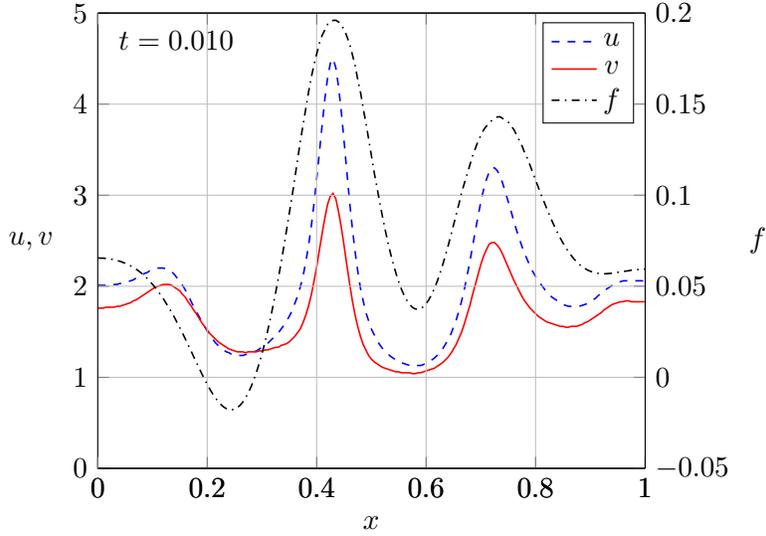

We can understand this dynamic as follows. Suppose that at time $t$, 
$u$ and $v$ each have a local maximum at a point $x^*\in \Omega$. Since $\nabla u= \nabla v=0$ at the point $(x^*,t)\in \Omega\times (0,\infty)$, the solution to \eqref{eq:reduced-form-PDE} locally obeys
\begin{align*}
  u_t(x^*,t) & = -u(x^*,t)\Delta f(x^*,t),\\
  v_t(x^*,t) & =-(1+\gamma)v(x^*,t) \Delta f(x^*,t).
\end{align*}
If $f(\cdot,t)$ also has a local maximum at $x^*$ and is such that
$\Delta f(x^*,t)\leq \Delta f(x,t)$ for $x$ in a neighborhood of
$x^*$, (for example if $f$ is well approximated by a quadratic in the
vicinity of its maxmimum), then
\begin{align*}
  u_t(x^*,t) & > u_t(x,t)\geq 0,\\
  v_t(x^*,t) & > v_t(x,t)\geq 0,
\end{align*}
 for $x$ near $x^*$.  The rate of increase at $x^*$ is greater than at
 nearby points, and the local maxima of $u$ and $v$ at $x^*$ remain at
$x^*$ at a later time $t+\delta t$.

Figure~\ref{solution1} demonstrates this typical evolution toward
a steady state, with  $\gamma = 0.5$.  The
initial conditions are
\begin{equation}
  \label{eq:example1-ic}
  \begin{aligned}
    u(x,0) &= 2 + \left(\frac{1}{5}\cos(3\pi x) + \frac{1}{2}\cos(5\pi
      x)\right)\exp\left(-\frac{(x-\tfrac{1}{2})^2}{x(1-x)}\right),\\
   v(x,0) & = 1.75 + \left(\frac{3}{10}\cos(2\pi x) + \frac{2}{5}
     \cos(4\pi x)\right)\exp \left(-\frac{(x-\tfrac{1}{2})^2}{x(1-x)}\right).
  \end{aligned}
\end{equation}
The factor $\exp\left(-\frac{(x-\tfrac{1}{2})^2}{x(1-x)}\right)$ is included to de-emphasize the role of the boundary, while satisfying the Neumann conditions. 
Early in the simulation ($t=0.010$, Figure~\ref{solution1}d), the
local extrema of $u$ and $v$ are aligned with one another, and
also aligned with the local extrema of the fitness function $f$ (see
Figure~\ref{fig:aligned-extrema}).  Evolution then progresses
asymptotically toward a steady
state where the fitness $f$ is constant and $u(x,t)=cv(x,t)$ throughout the domain $\Omega$. During this second phase, the aligned
maxima are increasing with time, while the aligned minima are
decreasing, but at a decreasing rate as the steady state is approached.

We have observed that for some initial conditions, $u$ and $v$ appear to 
actually reach
zero pointwise before the steady state is achieved (in finite time), leading
to the development of weak steady state solutions. These initial conditions seem to correlate with 
$uv\ll \abs{\nabla f}^2$ in some region of $\Omega$. In particular, by decreasing the
initial conditions in \eqref{eq:example1-ic} by a constant, we seem to be able to produce a weak steady state solution (see Figure~\ref{fig:comparison}).
Several examples are included in Section~\ref{weak-soln-example} below.

%We note that for some  initial conditions, namely when $uv\ll \abs{\nabla
%f}^2$ in some region of $\Omega$, then  $u$ and $v$ may actually reach
%zero at some points before the steady state is achieved, leading
%to weak steady-state solutions.  We give several examples is
%Section~\ref{weak-soln-example} below.   In particular, by decreasing the
%initial conditions in \eqref{eq:example1-ic} by a constant,  a weak
%steady state solution can be achieved (see Figure~\ref{fig:comparison}).

\subsection{Perturbation from Steady State}
\label{perturbation-example}

\begin{figure}[t]
  \centering
  \setlength{\figwidth}{.43\textwidth}
  \input{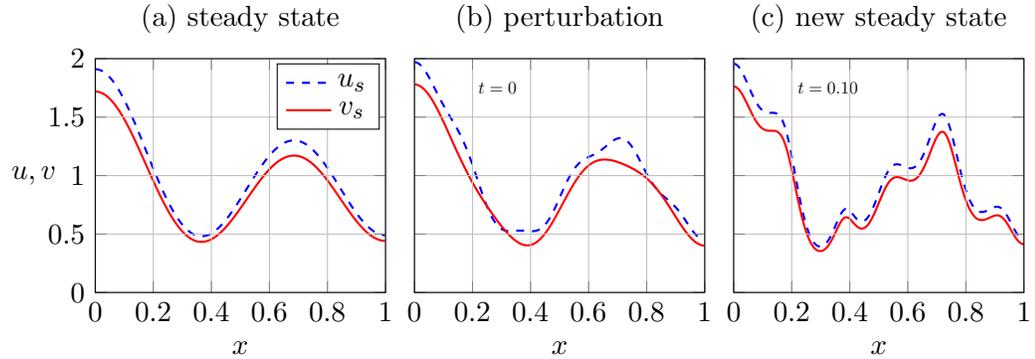} 
  \caption{(a) A steady state solution $(u_s,v_s)$ given by
    \eqref{eq:ss-solution}; (b) perturbation of $(u_s,v_s)$ given by
    \eqref{eq:perturbation-u}; (c) Evolution
    of perturbed problem to nearby steady state $(u_s',v_s')$.}
\label{perturbation}
\end{figure}

% \begin{figure}[t]
%   \centering
%   \setlength{\figwidth}{.7\textwidth}
%   \input{Fig6.tex}
%   %\includegraphics[width=.8\textwidth]{perturbation-v-linearization}
%   \caption{Comparison of $u_s'$, the steady state of the perturbed
%     problem, with the steady state of the linearized problem
%     $\hat{u}_s$.  The difference between the two steady states can be
%     attributed to nonlinear interactions of the eigenmodes in the full
%   problem that are not accounted for in the linearization.}
%   \label{linear-steady-state-example}
% \end{figure}

Figure~\ref{perturbation} illustrates the instability of 
smooth, strictly positive steady states; Figure~\ref{perturbation}a
shows a steady state solution $(u_s,v_s)$, where  
\begin{equation}
  \label{eq:ss-solution}
  \begin{aligned}
    u_s & = 1 + \frac{1}{5}\cos(\pi x)+\frac{1}{5}\cos(2\pi
    x)+\frac{1}{2}\cos(3\pi x) + \frac{1}{100}\cos(5\pi x),\\
    v_s &= \frac{9}{10}u_s.
  \end{aligned}
\end{equation}
This steady state is  perturbed at
$t=0$ (Figure~\ref{perturbation}b), 
\begin{equation}
  \label{eq:perturbation-u}
  \begin{aligned}
    u(x,0) & = u_s + \frac{1}{100}\cos(4\pi x)+\frac{1}{20}\cos(11\pi x),\\
    v(x,0) & = v_s + \frac{1}{100}\cos(2\pi x)+ \frac{1}{20}\cos(7\pi
    x).
  \end{aligned}
\end{equation}
With this perturbation as the initial condition, the solution to
\eqref{eq:reduced-form-PDE} evolves to a nearby steady state,
$(u_s',v_s')$, shown in Figure~\ref{perturbation}(c).

In Figure~\ref{fig:linearization-sim2} we show a simulation
for the linearization around a steady state $(u_s,v_s)$, where 
\begin{equation*}
  v_s(x) = \cos(2x)\exp(-(\pi-x)^2)+2, \quad \text{ and } \quad u_s(x) = 2v_s(x).
\end{equation*}
This steady-state $(u_s,v_s)$ is shown by the dashed plots in the
figure. The solid blue and black plots show the initial conditions for a perturbation $\mathbf{w}(x,0) =
(w_1(x,0),w_2(x,0))$, given by
\begin{equation*}
  w_1(x,0) = 0.5\cos(x) - 0.75\cos(4.5x), \quad w_2(x,0) = 0.5
  \cos(1.5x) + 0.3\cos(2.5x).
\end{equation*}
Computing the decomposition $\textbf{w}(x,t) = c_0(x)\mathbf{e}_0 +
c_1(x,t)\mathbf{e}_1$, we simulate a solution to the linearization
\eqref{eq:decoupled-pde} using the initial condition $c_1(x,0)$. The
plot for the final time (a steady state for the linearization) is also
shown (the green and orange plots in the figure).
\begin{equation*}
  w_1(x,T) = u_s(x)+cc_0(x) + \frac{c}{1+\gamma}c_1(x,T), \quad
  w_2(x,T) = v_s(x)+c_0(x) + c_1(x,T).
\end{equation*}

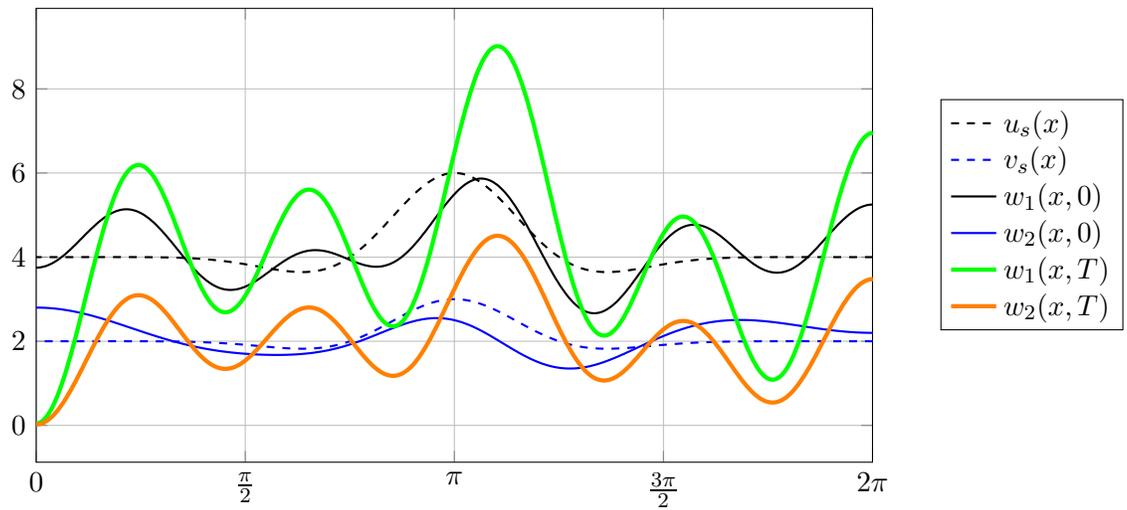
\begin{figure}
  \centering
  \begin{tikzpicture}
    \begin{axis}[width=5in,height=3in,grid=both,
      xmin=0,xmax=6.28,
      xtick={0,1.57,3.14,4.71,6.28},
      xticklabels = {0,$\frac{\pi}{2}$, $\pi$, $\frac{3\pi}{2}$,$2\pi$},
      legend style={at={(1.3,.8)}},
      legend cell align={left}]
      \addplot[black, thick,dashed] table [x=x, y=us, col sep=comma]
      {linear2.dat};
      \addplot[blue, thick,dashed] table [x=x, y=vs, col sep=comma]
      {linear2.dat};
      \addplot[black,thick,] table [x=x, y=w10, col sep=comma]
      {linear2.dat};
      \addplot[blue, thick,] table [x=x, y=w20, col sep=comma]
      {linear2.dat};

      \addplot[green, ultra thick] table [x=x, y=w1f, col sep=comma]
      {linear2.dat};
      \addplot[orange,ultra thick ] table [x=x, y=w2f, col sep=comma]
      {linear2.dat};
      \legend{$u_s(x)$,$v_s(x)$, ${w_1(x,0)}$, ${w_2(x,0)}$,
        ${w_1(x,T)}$, ${w_2(x,T)}$}

\end{axis}
\end{tikzpicture}
\caption{The dashed plots depict a steady state solution $(u_s,v_s)$,
  where $u_s(x) = 2v_s(x)$. The functions $w_1(x,0)$ and $w_2(x,0)$
  denote a perturbation of this steady state, which evolves according
  to the linearization given by \eqref{eq:linear-ffpde}. The
  functions $w_1(x,T)$ and $w_2(x,T)$ denote the steady-state that
  this perturbation evolves to. \label{fig:linearization-sim2}}
\end{figure}

% Let us compare this result with solutions of the linearized problem
% \eqref{linear-PDE2}.   By
% Theorem~\ref{thm:ss-linear}, the solution of \eqref{linear-PDE2}
% with initial conditions given by
% \eqref{eq:perturbation-u} approaches the steady state 
% \begin{equation}
%   \label{eq:ss-linear}
%   \begin{aligned}
%     \hat{u}_s & = u_s -\frac{1}{45}\cos(2\pi x)
%     +\frac{3}{100}\cos(4\pi x) -
%     \frac{1}{9}\cos(7\pi x) + \frac{3}{20}\cos(11\pi x),\\
%     \hat{v}_s & = \frac{9}{10}\hat{u}_s.
%   \end{aligned}
% \end{equation}
% In Figure~\ref{linear-steady-state-example}, we compare $\hat{u}_s$
% with the steady state solution for the full problem, $u_s'$ from Figure~\ref{perturbation}(c).  These two solutions are in close agreement; differences can be attributed to nonlinear
% interactions of the eigenmodes in the full problem, which are absent in the linearized problem.

\subsection{Evolution toward weak steady state solutions}
\label{weak-soln-example}

\begin{figure}[t]
  \centering
  \setlength{\figwidth}{.35\textwidth}
  \input{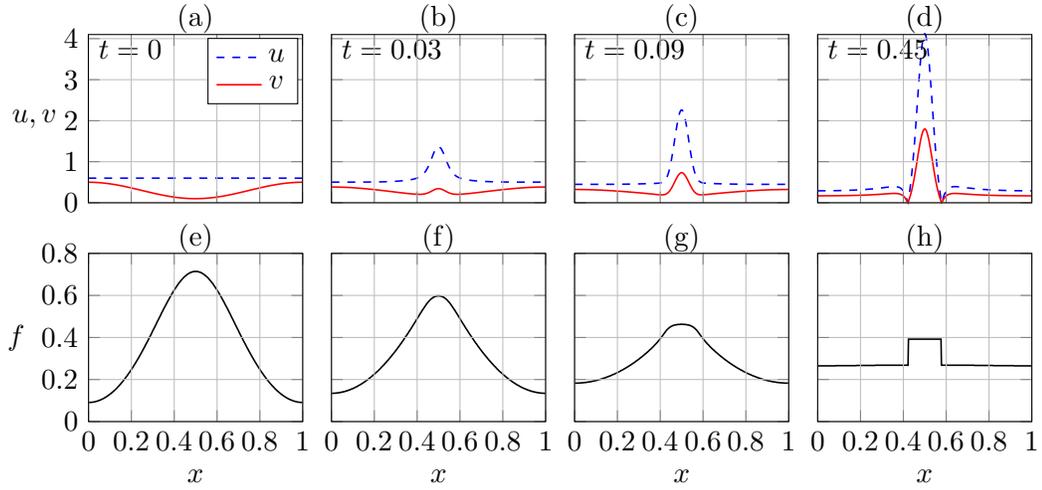}
  \caption{Evolution toward non-smooth steady-state (Example 1).  The
    fitness $f$ approaches a piecewise constant function.}
  \label{fig:weaksoln}
\end{figure}

As noted in Section~\ref{sec:ss-solns}, given certain initial
conditions, a solution of \eqref{eq:reduced-form-PDE} may evolve to a weak
steady state solution. These solutions are continuous but not
smooth, and the corresponding fitness function $f$ becomes
piecewise constant in the steady state.  We present several numerical examples.

It will be helpful to first discuss the implicit dynamics of $f$, the fitness
function for $u$.  Recall from Section~\ref{sec:recast} that $f$ depends only on
$u$ and $v$,
\begin{equation*}
  f(u,v) = \frac{a_{11}u+a_{12}v}{(u,v)},
\end{equation*}
where we have assumed that $a_{11}-a_{12}=1$.
We will use the notation $f(x,t)$ to refer to
$f(u(x,t),v(x,t))$.  Given a solution $(u,v)$ to
\eqref{eq:reduced-form-PDE}, notice that
\begin{equation}\label{eq:fitness-PDE}
f_t  = \frac{1}{(u+v)^2}(vu_t - uv_t),
\end{equation}
from which we can obtain
\begin{equation}
f_t  = \frac{1}{(u+v)^2}\left[\gamma uv\Delta f -
  \abs{\nabla f}^2 + \gamma u\nabla v\cdot \nabla f\right].  \label{eq:fitness-PDE} 
\end{equation}
%\begin{remark}
  If $u$ and $v$ are smooth strictly positive solutions, then the
  coefficient on $\Delta f$ in \eqref{eq:fitness-PDE} is positive, and
  it is clear from the maximum principle that $f$ attains its maximum
  and minimum values on the parabolic boundary, $\partial\Omega \times
  [t=0].$
%\end{remark}

\bigskip

\paragraph{Weak solution - Example 1.}
\label{ex:weaksoln1}

Let $(u(x,t), v(x,t))$ be a weak solution to \eqref{eq:reduced-form-PDE},
with the initial conditions 
\begin{equation}
  \label{eq:weaksoln-ic}
u(x,0) = \frac{3}{5},\quad v(x,0) =
\frac{1}{5}\cos(2\pi x) + \frac{3}{10}. 
\end{equation}
In Figure~\ref{fig:weaksoln}, the top row depicts the evolution of $u(x,t)$ and
$v(x,t)$ (Fig.~\ref{fig:weaksoln}a-d), while the bottom row depicts $f$, the
fitness of population $u$ (Fig.~\ref{fig:weaksoln}e-h). The initial conditions 
were chosen so that the fitness $f$ would have a single peak and no
interior minima, and such that
$v(x,0)$ is nearly zero over part of the domain.  Notice in
Figure~\ref{fig:weaksoln}c and g that the local minima of $u$ and $v$
are nearly aligned.  In the vicinity of these local minima we have the following conditions
\begin{enumerate}
\item $f_{xx}>0$,
\item $\abs{\nabla f}^2 >1$,
\item the product $uv\ll 1$.
\end{enumerate}
If $u(\cdot,t)$ has a local minimum at $x^*$, then
at the point $(x^*,t)$ we have
\begin{equation}
\label{eq:PDE-at-min}
  u_t(x^*,t) = -u(x^*,t)\Delta f(x^*,t).
\end{equation}
From the convexity of $f$, it is clear that $u$ and $v$ are decreasing in the vicinity
of the local minima.  If also, $\abs{\nabla f}^2> \gamma uv \Delta f$, then we
see from \eqref{eq:fitness-PDE} that $f$ will also be decreasing.
As a result, in the vicinity of the local minima of $u$ and $v$, both
$\Delta f$ and $\abs{\nabla f}$ are increasing.
Notice how this differs from the case where local minima of $u$ and
$v$ are aligned with a local minimum of  $f$.  We have
\begin{equation*}
  \Delta f = \frac{v\Delta u - u \Delta v}{(u+v)^3}>0, 
\end{equation*}
with $\Delta u$ and $\Delta v$ increasing.  Therefore
\begin{equation*}
  \Delta f \sim \frac{u \Delta v}{v^3},
\end{equation*}
If we assume that  $\Delta v$ is not decreasing and that $u\sim v$, we
conclude that $\Delta f \sim u^{-2}$ in the vicinity of $x^*$.

As long as the local minimum for $u$ remains at $x^*$ as $t$ increases,  then
\begin{equation*}
  u_t(x^*,t) = -u(x^*,t)\Delta f(x^*,t) \sim -u^{-1}(x^*,t)
\end{equation*}
Thus we expect $u \sim \sqrt{C-t}$ for some constant $C>0$, which
implies $u(x^*,t)$ goes to zero in finite time.

The dynamic here is one in which both $u$ and $v$ locally sense a high
fitness gradient, and their response has the effect of increasing this
gradient, thus accelerating the rate at which the densities $u$
and $v$ locally approach zero.  

\bigskip

\paragraph{Weak solution - Example 2.}
\label{ex:weaksoln2}

\begin{figure}[t]
  \centering
  \setlength{\figwidth}{.5\textwidth}
  \input{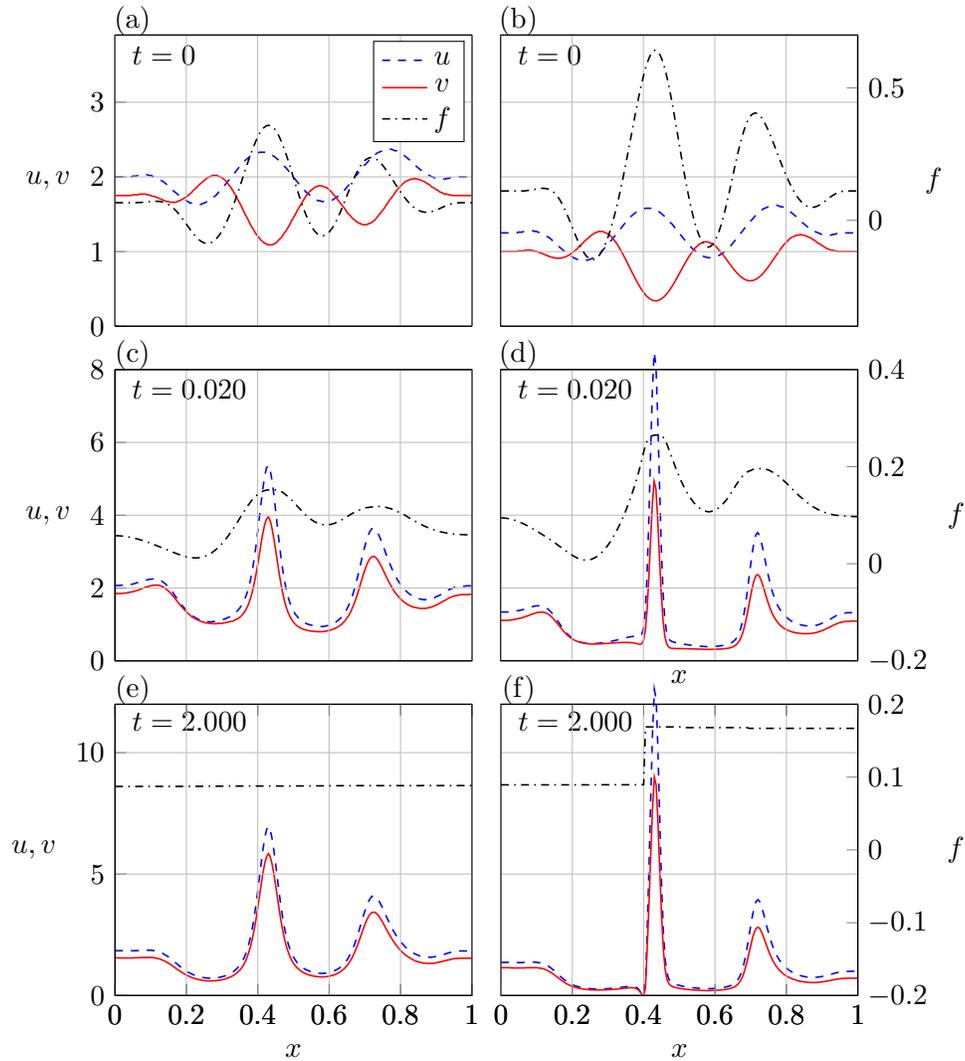}
   \caption{Comparison for two solutions whose initial conditions differ
    by a constant. In each plot scaling for densities $u$ and
    $v$, is indicated on the left, while scaling for $f(u,v)$ is
    indicated on the right. Left column shows evolution for (a) initial conditions $u_0,v_0$ from
    \eqref{eq:example1-ic}.  The right column shows evolution for (b) initial conditions are
    $u_0-\tfrac{3}{4},v_0-\tfrac{3}{4}$.  
    This produces large gradients in the
    fitness, along with regions where the product $uv\ll 1$ leading to
    a weak solution (see text).
  }
  \label{fig:comparison}
\end{figure}

In our second example, we modify the initial
conditions \eqref{eq:example1-ic} from the example in
Section~\ref{solution-example}, subtracting the constant $1$ from
each initial condition (see Figure~\ref{fig:comparison}(a),(b)).  This changes the
relative values of $u$ and $v$, thus altering the fitness profile and setting up the condition, $uv\ll 1$ on part of the domain, that leads to
weak solutions %in the previous example
(Figure~\ref{fig:comparison}(c),(d)).  In the steady state
(Figure~\ref{fig:comparison}(e),(f)), the fitness profile for the modified
problem is piecewise constant. 

\bigskip

\paragraph{Weak solution - Example 3.}
\label{ex:weak-soln3}

\begin{figure}[t]
  \centering
  \setlength{\figwidth}{.55\textwidth}
  \input{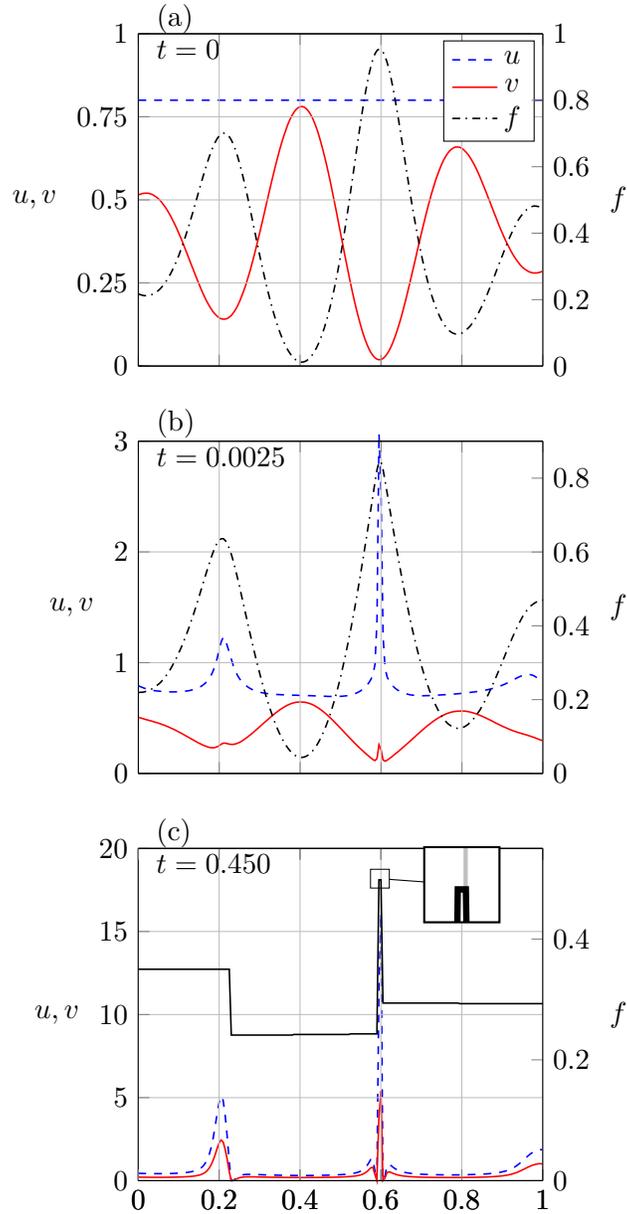}
  \caption{Example 3: (a) initial conditions that produce large
    gradients in the
    fitness and regions where $uv\ll 1$ (near $x=0.6$), leading to a weak steady
    state solution with a fitness that is piecewise constant in the
    steady state (c). Note that the fitness is constant on a small
    interval around $x=0.6$ (see inset).
    }
  \label{fig:weak-soln3}
\end{figure}

Our third weak solution is shown in Figure~\ref{fig:weak-soln3}, with 
initial conditions
\begin{equation}
  \label{eq:weak-soln3-ic}
  \begin{aligned}
    u(x,0) &= \frac{3}{5},\\
    v(x,0) & = \frac{2}{5}+\left(\frac{2}{5}\cos(5\pi x)\right)\exp\left(-\frac{(x-\tfrac{1}{2})^2}{x(1-x)}\right).
  \end{aligned}
\end{equation}
Notice from Figure~\ref{fig:weak-soln3}(a),(d) that at time $t=0$, the
local minima of $v$ correspond to local maxima of $f$, at the points
$x^*$ and $y^*$ in the figure. The local maxima of $f$ drive
aggegration of $u$ and $v$ in the vicinity of $x^*$ and $y^*$
(Figure~\ref{fig:weak-soln3}(b)), which in turn leads to  local
minima in $u$ and $v$ near $x^*$ and $y^*$, with large gradients in
$f$, the conditions that drive $u$ and $v$ to zero.

%\begin{remark}
%  
%\end{remark}

\subsection{A spatial Lotka--Volterra model}
\label{sec:lv}

In our final example, we use the fitness gradient flux to construct a
spatial Lotka--Volterra model.  The non-spatial model has a
stable steady state, which we show here to be destabilized by cross-diffusion when the
fitness gradient flux is included.  Our approach will be discussed
more fully in a future paper.

Consider a generalized Lotka--Volterra ODE model
\begin{equation}
  \begin{aligned}
    \label{lv-ode}
    \dot{u} & = g_1(u,v) = u(c_1-c_2u-c_3v),\\
    \dot{v} & = g_2(u,v) = v(k_1+k_2u-k_3v),
  \end{aligned}
\end{equation}
where $u$ and $v$ are densities of the two species subject to logistic growth, 
and the constants $c_i,k_i> 0$.   The growth rate of $v$ is enhanced
by $u$, while the growth rate of $u$ is decreased by $v$, as might occur
in a predator-prey or host-parasite type interaction, where $u$ is the
prey and $v$ is the predator. We note however
that in the standard Lotka--Volterra predator-prey system, the constant $k_1$
would be strictly negative.   If the null-clines
$c_1=c_2u-c_3v$ and $k_1 = -k_2u+k_3v$ in the $(u,v)$-phase plane intersect in the interior of
the first quadrant, then \eqref{lv-ode} has a stable steady state
$(u^*,v^*)$, with $u^*,v^*>0$.

Linearizing \eqref{lv-ode} around the steady state and letting
\begin{equation*}
  J_{(u^*,v^*)} =
  \begin{bmatrix}
    \partial_u g_1 & \partial_v g_1\\ \partial_u g_2 & \partial_v g_2
  \end{bmatrix}_{(u^*,v^*)} =
  \begin{bmatrix}
    -k_2u^* & -k_3 u*\\ c_2v^* & -c_3v^*
  \end{bmatrix},
\end{equation*}
we see that $J$ has the following sign structure:
\begin{equation}
  \label{Jsign}
  \begin{bmatrix}
    - & - \\ + & -
  \end{bmatrix}.
\end{equation}
Since $\partial_u g_1$ and $\partial_v g_2$ have the same sign at the
steady state, the ODE system does not display an activator-inhibitor
dynamic, and  the steady state cannot be destabilized by
diffusion \cite{Murray2002}.  However, a cross-diffusive instability occurs when we
spatially extend this model as a fitness gradient flux system
\begin{equation}
  \label{eq:lv-pde}
  \begin{aligned}
    \partial_t u & = -\nabla \cdot (u \nabla f) + u(c_1 -c_2u -c_3v),\\
    \partial_t v & = -(1+ \gamma)\nabla \cdot (v\nabla g) + v(k_1 +k_2u-k_3v),
  \end{aligned}
\end{equation}
where the fitness $f$ is as defined in Section~\ref{sec:recast} above,
and satisfies the same conditions as assumed in our previous analysis. 
The instability is illustrated with a numerical example in
Figure~\ref{fig:LV}. Note that an individual in either the prey or
predator population
benefits by locating itself where there is a high density of prey
relative to predators.  Prey tend to aggregate, and predators
follow. The result is an alignment dynamic for the extrema as shown in
Figure \ref{fig:LV}, very similar to the fitness gradient flux system
discussed previously.  
Spatial variation of the ratio $u/v$ in the initial
conditions give rise to local aggregations, as the populations align
in a spatially structured steady state.  Unlike Turing patterns, however, 
there is no characteristic wavelength; 
steady state patterns depend on initial conditions.

\begin{figure}[t]     
  \centering
  \setlength{\figwidth}{.43\textwidth}
  \input{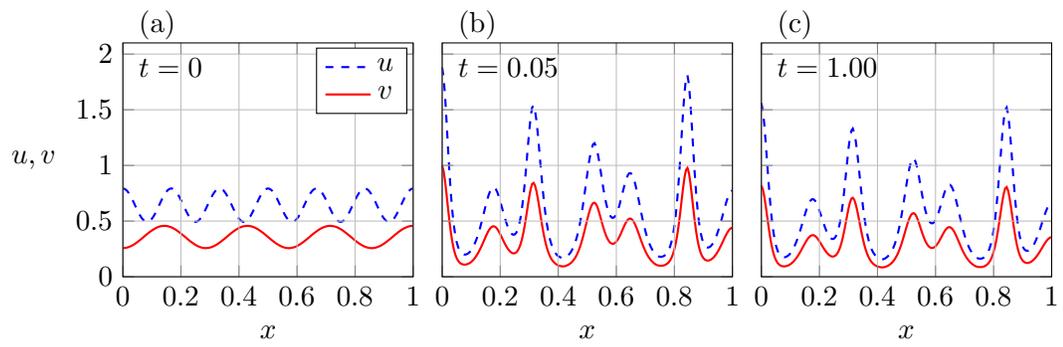}
  \caption{Spatially-extended Lotka--Volterra model \protect \eqref{eq:lv-pde}, 
  with $k_1=k_2=k_3=1, c_1=3/2, c_2=1, c_3=6$.}
  \label{fig:LV}
\end{figure}

%%% Local Variables: 
%%% mode: latex
%%% TeX-master: "GameTheoreticAggregationMain"
%%% End: 

%%%%%%%%%%%%%%%%%%%%%%%%%%%%%%%%%%%%%%%%%%%%%%%%%%%%%%%

\section{Conclusions}
\label{sec:conclusion}

Our results show that under a fitness-based dispersal mechanism where
the fitness has some dependence on individual interactions, as in an evolutionary game, variations in the ratio of population
densities lead to spatial structure as populations ascend local
fitness gradients.  

The interaction between populations in our model has a predator-prey or
cooperative-exploitative dynamic, as in the standard prisoner's dilemma and
hawk-dove games.  Individuals of both populations benefit
by locating themselves where the density of the cooperative or prey
species $u$ is large, relative to the density of the exploitative or predatory species
$v$.

We can consider interesting extensions of the model by coupling this fitness gradient
flux with ODE systems for relevant local population dynamics, as we
have done in the spatial Lotka--Volterra model in
Section~\ref{sec:lv}.  We also expect this spatial coupling to have
relevance to public-goods interactions that describe coexistence of
cooperative and exploitative behavior as has been observed, for example,
in polymorphic populations of yeast \cite{Greig2003}. 

Although here we have focused on directed motion in a non-diffusive
limit, it is natural to consider including a component of diffusion
and/or a law of motion for each population in the absence of the
other, as well as a density dependent fitness or term describing
interactions when $u+v$ is small and the mean-field assumptions of
evolutionary game theory should not be expected to hold. We also have not included in our basic model any term that
\emph{a priori} prevents unlimited aggregation.  The higher
sensitivity of population $v$,  to the
fitness gradient ($\gamma>0$ in
\eqref{eq:reduced-form-PDE}) allows the exploitative population to in some sense
overtake $u$ and limit its aggregation.  We expect that if the
cooperative population $u$ has the higher sensitivity that blow up would
occur, although this remains to be shown.

It is also  interesting to consider 
non-transitive (cyclic) games for three players, such as the classic Rock, Paper,
Scissors game.  Through numerical simulations,  we have shown the
development of spiral waves in 2D in a previous paper, and we suspect
that such models also have periodic solutions when coupled with particular
local population dynamics \cite{deforest2013}.

%%%%%%%%%%%%%%%%%%%%%%%%%%%%%%%%%%%%%%%%%%%%%%%%%%%%%%%

%%% Local Variables: 
%%% mode: latex
%%% TeX-master: "GameTheoreticAggregationMain"
%%% End: 

% \begin{acknowledgements} 
\noindent\textbf{Acknowledgements:} We would like to thank H. K. Jenssen and Y. Lou
 for helpful discussions, and C. Cosner and T. Reluga for comments. AB
 was supported by NSF Grant CMMI-1463482.
% \end{acknowledgements}
% \bibliographystyle{spmpscinat}
% \bibliographystyle{plainnat}
\raggedright
%\bibliography{FFbib}

\begin{thebibliography}{41}
\providecommand{\natexlab}[1]{#1}
\providecommand{\url}[1]{\texttt{#1}}
\expandafter\ifx\csname urlstyle\endcsname\relax
  \providecommand{\doi}[1]{doi: #1}\else
  \providecommand{\doi}{doi: \begingroup \urlstyle{rm}\Url}\fi

\bibitem[Amann(1988)]{Amann1988b}
Herbert Amann.
\newblock {Dynamic theory of quasilinear parabolic equations I. Abstract
  evolution equations}.
\newblock \emph{Nonlinear Analysis: Theory, Methods \& Applications},
  12\penalty0 (9):\penalty0 895--919, September 1988.
\newblock \doi{10.1016/0362-546X(88)90073-9}.

\bibitem[Amann(1989)]{Amann1989b}
Herbert Amann.
\newblock {Dynamic Theory of Quasilinear Parabolic Systems: III. Global
  Existence}.
\newblock \emph{Math. Z.}, 202:\penalty0 219--250, 1989.
\newblock \doi{10.1007/BF02571246}.

\bibitem[Amann(1990)]{AmannQ2}
Herbert Amann.
\newblock {Dynamic Theory of Quasilinear Parabolic Equations II.
  Reaction-Diffusion Systems}.
\newblock \emph{Differential and Integral Equations}, 3\penalty0 (1):\penalty0
  13--75, 1990.
\newblock URL \url{https://projecteuclid.org/euclid.die/1371586185}.

\bibitem[Aronson(1985)]{Aronson1985}
D.~G. Aronson.
\newblock The role of diffusion in mathematical population biology: {S}kellam
  revisited.
\newblock In \emph{Mathematics in biology and medicine ({B}ari, 1983)},
  volume~57 of \emph{Lecture Notes in Biomath.}, pages 2--6. Springer, Berlin,
  1985.
\newblock \doi{10.1007/978-3-642-93287-8_1}.

\bibitem[Baca{\"e}r(2011)]{Bacaer2011}
N.~Baca{\"e}r.
\newblock \emph{A Short History of Mathematical Population Dynamics}.
\newblock Springer, 2011.
\newblock \doi{10.1007/978-0-85729-115-8}.

\bibitem[Bedrossian et~al.(2011)Bedrossian, Rodr\'{i}guez, and
  Bertozzi]{Bedrossian2011}
Jacob Bedrossian, Nancy Rodr\'{i}guez, and Andrea~L. Bertozzi.
\newblock {Local and global well-posedness for aggregation equations and
  Patlak-Keller-Segel models with degenerate diffusion}.
\newblock \emph{Nonlinearity}, 24:\penalty0 1683--1714, 2011.
\newblock \doi{10.1088/0951-7715/24/6/001}.

\bibitem[Bertozzi and Slepcev(2010)]{Bertozzi2010}
Andrea~L. Bertozzi and Dejan Slepcev.
\newblock {Existence and Uniqueness of Solutions to an Aggregation Equation
  with Degenerate Diffusion}.
\newblock \emph{Communications on Pure and Applied Analysis}, 9\penalty0
  (6):\penalty0 1617--1637, 2010.
\newblock \doi{10.3934/cpaa.2010.9.1617}.

\bibitem[Cantrell and Cosner(2003)]{Cantrell2003}
Robert~Stephen Cantrell and Chris Cosner.
\newblock \emph{Spatial ecology via reaction-diffusion equations}.
\newblock Wiley Series in Mathematical and Computational Biology. John Wiley \&
  Sons, Ltd., Chichester, 2003.
\newblock ISBN 0-471-49301-5.
\newblock \doi{10.1002/0470871296}.

\bibitem[Cantrell et~al.(2008)Cantrell, Cosner, and Lou]{Cantrell2008}
Robert~Stephen Cantrell, Chris Cosner, and Yuan Lou.
\newblock Approximating the ideal free distribution via
  reaction-diffusion-advection equations.
\newblock \emph{J. Differential Equations}, 245\penalty0 (12):\penalty0
  3687--3703, 2008.
\newblock \doi{10.1016/j.jde.2008.07.024}.

\bibitem[Cantrell et~al.(2013)Cantrell, Cosner, Lou, and Xie]{Cantrell2013}
Robert~Stephen Cantrell, Chris Cosner, Yuan Lou, and Chao Xie.
\newblock Random dispersal versus fitness-dependent dispersal.
\newblock \emph{J. Differential Equations}, 254\penalty0 (7):\penalty0
  2905--2941, 2013.
\newblock \doi{10.1016/j.jde.2013.01.012}.

\bibitem[Childress and Percus(1981)]{Childress1981}
S.~Childress and J.~K. Percus.
\newblock Nonlinear aspects of chemotaxis.
\newblock \emph{Mathematical Biosciences}, 56\penalty0 (4):\penalty0 217--237,
  1981.
\newblock \doi{10.1016/0025-5564(81)90055-9}.

\bibitem[Cosner(2005)]{Cosner2005}
Chris Cosner.
\newblock A dynamic model for the ideal-free distribution as a partial
  differential equation.
\newblock \emph{Theoretical Population Biology}, 67\penalty0 (2):\penalty0
  101--108, 2005.
\newblock \doi{10.1016/j.tpb.2004.09.002}.

\bibitem[Cosner(2014)]{Cosner2014}
Chris Cosner.
\newblock Reaction-diffusion-advection models for the effects and evolution of
  dispersal.
\newblock \emph{Discrete Contin. Dyn. Syst.}, 34\penalty0 (5):\penalty0
  1701--1745, 2014.
\newblock \doi{10.3934/dcds.2014.34.1701}.

\bibitem[Cressman and Křivan(2006)]{Cressman2006}
Ross Cressman and Vlastimil Křivan.
\newblock Migration dynamics for the ideal free distribution.
\newblock \emph{The American Naturalist}, 168\penalty0 (3):\penalty0 pp.
  384--397, 2006.
\newblock \doi{10.1086/506970}.

\bibitem[DeForest and Belmonte(2013)]{deforest2013}
Russ DeForest and Andrew Belmonte.
\newblock {Spatial pattern dynamics due to the fitness gradient flux in
  evolutionary games}.
\newblock \emph{Physical Review E}, 87\penalty0 (6):\penalty0 062138, June
  2013.
\newblock \doi{10.1103/PhysRevE.87.062138}.

\bibitem[Demetrius and Gundlach(2000)]{Demetrius2000}
Lloyd Demetrius and Volker~Matthais Gundlach.
\newblock Game theory and evolution: finite size and absolute fitness measures.
\newblock \emph{Mathematical Biosciences}, 168:\penalty0 9--38, 2000.
\newblock \doi{10.1016/S0025-5564(00)00042-0}.

\bibitem[Easley and Kleinberg(2010)]{Easley2010}
David Easley and Jon Kleinberg.
\newblock \emph{{Networks, Crowds, and Markets: Reasoning about a Highly
  Connected World}}.
\newblock Cambridge University Press, 2010.
\newblock URL \url{http://www.cs.cornell.edu/home/kleinberg/networks-book}.

\bibitem[Fisher(1937)]{Fisher1937}
R.~A. Fisher.
\newblock The wave of advance of advantageous genes.
\newblock \emph{Annals of Eugenics}, 7\penalty0 (4):\penalty0 355--369, 1937.
\newblock \doi{10.1111/j.1469-1809.1937.tb02153.x}.

\bibitem[Fretwell and Jr.(1970)]{Fretwell1970}
Stephen~D. Fretwell and Henry L.~Lucas Jr.
\newblock {On territorial behavior and other factors influencing habitat
  distribuion in birds. I. Theoretical development}.
\newblock \emph{Acta Biotheoretica}, 14:\penalty0 16--36, 1970.
\newblock ISSN 0001-5342.
\newblock \doi{10.1007/BF01601953}.

\bibitem[Greig and Travisano(2004)]{Greig2003}
Duncan Greig and Michael Travisano.
\newblock The prisoner{\textquoteright}s dilemma and polymorphism in yeast
  {SUC} genes.
\newblock \emph{Proceedings of the Royal Society of London B: Biological
  Sciences}, 271\penalty0 (Suppl 3):\penalty0 S25--S26, 2004.
\newblock ISSN 0962-8452.
\newblock \doi{10.1098/rsbl.2003.0083}.

\bibitem[Hamilton(1971)]{Hamilton1971}
W.~D. Hamilton.
\newblock {Geometry for the Selfish Herd}.
\newblock \emph{J.~theor.~Biol.}, 31:\penalty0 295--311, 1971.
\newblock \doi{10.1016/0022-5193(71)90189-5}.

\bibitem[Hofbauer and Sigmund(1998)]{Hofbauer1998}
J.~Hofbauer and K.~Sigmund.
\newblock \emph{{Evolutionary Games and Population Dynamics}}.
\newblock Cambridge University Press, 1998.
\newblock ISBN 9780521625708.

\bibitem[Horstmann(2003)]{Horstmann2003}
Dirk Horstmann.
\newblock From 1970 until present: the {Keller-Segel} model in chemotaxis and
  its consequences, 2003.
\newblock URL \url{http://www.mis.mpg.de/preprints/2003/preprint2003_3.pdf}.
\newblock Preprint.

\bibitem[Laurent(2007)]{Laurent2007}
Thomas Laurent.
\newblock {Local and Global Existence for an Aggregation Equation}.
\newblock \emph{Communications in Partial Differential Equations}, 32:\penalty0
  1941--1964, 2007.
\newblock \doi{10.1080/03605300701318955}.

\bibitem[Lui(2011)]{Lui2011}
S.~H. Lui.
\newblock \emph{{Numerical Analysis of Partial Differential Equations}}.
\newblock John Wiley \& Sons, 2011.
\newblock ISBN 978-1-118-11113-0.

\bibitem[Morisita(1971)]{Morisita1971}
M.~Morisita.
\newblock Measuring of habitat value by environmental density method.
\newblock In GP~Patil, EC~Pielou, and WE~Waters, editors, \emph{Statistical
  Ecology Vol. I}, pages 379--401. The Pennsylvania State University Press,
  1971.

\bibitem[Morrell and James(2008)]{Morrell2007}
Lesley~J. Morrell and Richard James.
\newblock Mechanisms for aggregration in animals: rule success depends on
  ecological variables.
\newblock \emph{Behavioral Ecology}, 19\penalty0 (1):\penalty0 193--201, 2008.
\newblock \doi{10.1093/beheco/arm122}.

\bibitem[Murray(2002)]{Murray2002}
J.~D. Murray.
\newblock \emph{Mathematical biology. {I}}, volume~17 of
  \emph{Interdisciplinary Applied Mathematics}.
\newblock Springer-Verlag, New York, third edition, 2002.
\newblock ISBN 0-387-95223-3.
\newblock An introduction.

\bibitem[Okubo(1986)]{Okubo1986}
Akira Okubo.
\newblock Dynamical aspects of animal grouping: Swarms, schools, flocks, and
  herds.
\newblock \emph{Advances in Biophysics}, 22:\penalty0 1--94, 1986.
\newblock \doi{10.1016/0065-227X(86)90003-1}.

\bibitem[Okubo and Levin(2001)]{Okubo2001}
Akira Okubo and Simon~A. Levin.
\newblock \emph{Diffusion and ecological problems: modern perspectives},
  volume~14 of \emph{Interdisciplinary Applied Mathematics}.
\newblock Springer-Verlag, New York, second edition, 2001.
\newblock ISBN 0-387-98676-6.
\newblock \doi{10.1007/978-1-4757-4978-6}.

\bibitem[Parrish et~al.(2002)Parrish, Viscido, and Gr\"{u}nbaum]{Parrish2002}
Julia~T. Parrish, Steven~V. Viscido, and Daniel Gr\"{u}nbaum.
\newblock {Self-Organized Fish Schools: An Examination of Emergent Properties}.
\newblock \emph{Biological Bulletin}, 202\penalty0 (3):\penalty0 296--305,
  2002.
\newblock \doi{10.2307/1543482}.

\bibitem[Rosenzweig and Abramsky(1985)]{Rosenzweig1985}
Michael~L. Rosenzweig and Zvika Abramsky.
\newblock Detecting density-dependent habitat selection.
\newblock \emph{The American Naturalist}, 126\penalty0 (3):\penalty0 pp.
  405--417, 1985.
\newblock URL \url{http://www.jstor.org/stable/2461364}.

\bibitem[Rowell(2010)]{Rowell2010}
Jonathan~T. Rowell.
\newblock Tactical population movements and distributions for ideally motivated
  competitors.
\newblock \emph{The American Naturalist}, 176\penalty0 (5):\penalty0 pp.
  638--650, 2010.
\newblock \doi{10.1086/656494}.

\bibitem[Ruxton and Sherratt(2006)]{Ruxton2006}
Graeme~D. Ruxton and Thomas~N. Sherratt.
\newblock Aggregation, defence and warning signals: The evolutionary
  relationship.
\newblock \emph{Proceedings: Biological Sciences}, 273\penalty0
  (1600):\penalty0 pp. 2417--2424, 2006.
\newblock URL \url{http://www.jstor.org/stable/25223621}.

\bibitem[Shigesada et~al.(1979)Shigesada, Kawasaki, and
  Teramoto]{Shigesada1979}
N.~Shigesada, K.~Kawasaki, and E.~Teramoto.
\newblock {Spatial segregation of interacting species.}
\newblock \emph{Journal of theoretical biology}, 79\penalty0 (1):\penalty0
  83--99, July 1979.
\newblock \doi{10.1016/0022-5193(79)90258-3}.

\bibitem[Skellam(1973)]{Skellam1973}
J.~G. Skellam.
\newblock The formulation and interpretation of mathematical models of
  diffusionary processes in population biology.
\newblock In M.~S. Bartlett and R.~W. Hiors, editors, \emph{{The Mathematical
  Theory of the Dynamics of Biological Populations}}, pages 63--85. Academic
  Press, London, 1973.

\bibitem[Skellam(1991)]{Skellam1951}
J.G. Skellam.
\newblock {Random Dispersal in Theoretical Populations}.
\newblock \emph{Bulletin of Mathematical Biology}, pages 135--165, 1991.
\newblock \doi{10.1007/BF02464427}.
\newblock Reprinted from Biometrika, \textbf{38}, 1951.

\bibitem[Taylor and Jonker(1978)]{Taylor1978}
Peter~D. Taylor and Leo~B. Jonker.
\newblock Evolutionarily stable strategies and game dynamics.
\newblock \emph{Math. Biosci.}, 40\penalty0 (1-2):\penalty0 145--156, 1978.
\newblock \doi{10.1016/0025-5564(78)90077-9}.

\bibitem[Vickers(1989)]{Vickers1989}
G.~T. Vickers.
\newblock Spatial patterns and {ESS}'s.
\newblock \emph{J. Theoret. Biol.}, 140\penalty0 (1):\penalty0 129--135, 1989.
\newblock \doi{10.1016/S0022-5193(89)80033-5}.

\bibitem[Wrona and Dixon(1991)]{Wrona1991}
Frederick~J. Wrona and R.~W.~Jamieson Dixon.
\newblock Group size and predation risk: A field analysis of encounter and
  dilution effects.
\newblock \emph{The American Naturalist}, 137\penalty0 (2):\penalty0 pp.
  186--201, 1991.
\newblock URL \url{http://www.jstor.org/stable/2462112}.

\bibitem[Xu et~al.(2017)Xu, Belmonte, deForest, Liu, and Tan]{Xu2017}
Qiuju Xu, Andrew Belmonte, Russ deForest, Chun Liu, and Zhong Tan.
\newblock Strong solutions and instability for the fitness gradient system in
  evolution games between two populations.
\newblock \emph{Journal of Differential Equations}, 262:\penalty0 4021--4051,
  2017.
\newblock \doi{10.1016/j.jde.2016.12.008}.

\end{thebibliography}

\end{document}